\input ./style/arxiv-general.cfg
\documentclass[MSNbibl,number,citesort,dvips]{arxbj}
\makeatletter
   \@ifpackageloaded{graphicx}{}{\usepackage{graphicx}}
\makeatother

\setvaluelist{emailmarks}{*,**,\textdagger,\textdaggerdbl,\S,\P,$\|$}

\volume{22}
\issue{4}
\pubyear{2016}
\firstpage{2113}
\lastpage{2142}
\doi{10.3150/15-BEJ722}
\docsubty{FLA}

\makeatletter
\newcommand{\lleft}{\left}
\newcommand{\rrvert}{\vert}
\newcommand{\rright}{\right}
\newcommand{\llvert}{\vert}

\def\nicefrac{\frac}

\newtheorem{Lemm}{Lemma}[section]

\newtheorem{Lemmm}{Lemma}
\newtheorem{Th}[Def]{Theorem}

\newtheorem{Thhh}{Theorem}
\newtheorem{Thhhh}{Theorem}

\newtheorem{Thhhhhh}{Theorem}

\newtheorem{thmmmmm}{Theorem}
\newremark{Rm}[Def]{Remark}
\newremark{rem}[Def]{Remark}
\newproclaim{Def}{Definition}[section]
\newproclaim{defn}[Def]{Definition}
\newtheorem{cor}[Def]{Corollary}

\newtheorem{assumption}[Def]{Assumption}

\newcommand{\var}{\operatorname{Var}}
\newcommand{\sign}{\operatorname{sign}}
\newcommand{\lw}{\ell}
\newtheorem{thmm}[Def]{Theorem}

\def\CC{\mathbb{C}}
\def\SS{\mathbb{S}}
\def\RR{\mathbb{R}}
\def\ra{\rightarrow}
\def\z{\mathbf{z}}
\def\ZZ{\mathbb{Z}}
\def\sgn{\operatorname{sgn}}
\def\var{\operatorname{var}}
\def\x{\mathbf{x}}
\def\ve{\varepsilon}
\def\cv{\mathfrak{c}}
\def\ssp{\mathrm{ssp}}
\makeatother

\begin{document}
\begin{frontmatter}

\title{The circular SiZer, inferred persistence of shape parameters and
application to early stem cell differentiation}
\runtitle{The WiZer}

\begin{aug}
\author[a]{\inits{S.}\fnms{Stephan} \snm{Huckemann}\corref{}\thanksref{a,e1}\ead[label=e1,mark]{huckeman@math.uni-goettingen.de}},
\author[b1]{\inits{K.-R.}\fnms{Kwang-Rae} \snm{Kim}\thanksref{b1,e2}\ead[label=e2,mark]{Kwang-rae.Kim@nottingham.ac.uk}},
\author[b2]{\inits{A.}\fnms{Axel} \snm{Munk}\thanksref{b2,e3}\ead[label=e3,mark]{munk@math.uni-goettingen.de}},
\author[b3]{\inits{F.}\fnms{Florian} \snm{Rehfeldt}\thanksref{b3,e4}\ead[label=e4,mark]{rehfeldt@physik3.gwdg.de}},
\author[a]{\inits{M.}\fnms{Max} \snm{Sommerfeld}\thanksref{a,e5}\ead[label=e5,mark]{max.sommerfeld@math.uni-goettingen.de}},
\author[b5]{\inits{J.}\fnms{Joachim}~\snm{Weickert}\thanksref{b5,e6}\ead[label=e6,mark]{weickert@mia.uni-saarland.de}}
\and
\author[b3]{\inits{C.}\fnms{Carina} \snm{Wollnik}\thanksref{b3,e7}\ead[label=e7,mark]{carina.wollnik@phys.uni-goettingen.de}}
\address[a]{Felix Bernstein Institute for Mathematical Statistics in
the Biosciences, University of G\"ottingen.\\ \printead{e1,e5}}
\address[b1]{School of Mathematical Sciences,\hspace*{-0.5pt} University of Nottingham.\hspace*{-0.5pt}\\
\printead{e2}}
\address[b2]{Max Planck Institute for Biophysical Chemistry, G\"
ottingen and Felix Bernstein Institute for Mathematical Statistics in
the Biosciences, University of G\"ottingen.\\ \printead{e3}}
\address[b3]{3rd Institute of Physics -- Biophysics, University of G\"
ottingen.\\ \printead{e4,e7}}
\address[b5]{Faculty of Mathematics and Computer Science, Saarland
University.\\ \printead{e6}}\vspace*{-3pt}
\runauthor{S. Huckemann et al.}
%
\end{aug}

%
\received{\smonth{4} \syear{2014}}
%
\revised{\smonth{11} \syear{2014}}

%
\begin{abstract}
We generalize the SiZer of Chaudhuri and Marron
(\textit{J. Amer. Statist. Assoc.} \textbf{94} (1999) 807--823;
\textit{Ann. Statist.} \textbf{28} (2000) 408--428) for the
detection of shape parameters of densities on the real line to the case
of circular data. It turns out that only the wrapped Gaussian kernel
gives a
symmetric, strongly Lipschitz semi-group satisfying ``circular''
causality, that is, not introducing possibly artificial modes with
increasing levels of smoothing. Some notable differences between
Euclidean and circular scale space theory are highlighted. Based on
this, we provide an asymptotic theory to make
inference about the persistence of shape features. The resulting
circular mode persistence diagram is applied to the analysis of early
mechanically-induced differentiation in adult human stem cells from
their actin-myosin filament structure. As a consequence, the circular
SiZer based on the wrapped Gaussian kernel (WiZer) allows the
verification at a controlled error level of the observation reported by
Zemel \textit{et~al.}
(\textit{Nat. Phys.} \textbf{6} (2010) 468--473):
Within early stem cell differentiation,
polarizations of stem cells exhibit preferred directions in three
different micro-environments.
\end{abstract}


\begin{keyword}
\kwd{circular data}
\kwd{circular scale spaces}
\kwd{mode hunting}
\kwd{multiscale process}
\kwd{persistence inference}
\kwd{stem cell differentiation}
\kwd{variation diminishing}
\kwd{wrapped Gaussian kernel estimator}
\end{keyword}
\end{frontmatter}

\section{Introduction}\label{sec1}

Mode (maxima and minima) and bump (maxima of the derivative) hunting of
a density has a long history in statistical research and has been
tackled from various perspectives.
Good and Gaskins \cite{GoodGaskins1980} argued that actually this
would be a problem of
significance testing rather than estimation, and, to some extent, we
agree. We stress, however, that
for practical purposes, it seems attractive to accompany any testing
decision on the number of modes, say, with an estimator and
corresponding visualization tools which
are in agreement as much as possible with such a test decision.
Indeed, many tests which have been developed (mainly in the context of
detecting modes of a density) implicitly or explicitly offer this
additional information or parts of it.\vadjust{\goodbreak}
Most of them are based on smoothing techniques with variable bandwidth
which provides a reconstruction of modes and other shape features at a
range of scales. Prominent methods include the critical bandwidth test
of Silverman \cite{Silverman1981} (see also Ahmed and Walther \cite
{AhmedWalter2012} for a
generalization to multivariate data), the dip test by Hartigan and
Hartigan \cite
{HartiganHartigan1985}, the excess mass approach by M{\"u}ller and
Sawitzki \cite
{MullerSawitzki1991}, see also Polonik \cite{Polonik1995}, the test of
Cheng and Hall \cite
{ChengHall1999}, the SiZer (SIgnificant ZERo crossings of the
derivatives) by Chaudhuri and Marron \cite{ChaudhuriMarron1999,ChaudhuriMarron2000}, or the
mode tree of Minnotte and Scott \cite{MinnotteScott1993}; see also
Minnotte \cite{Minnotte1997},
Ooi \cite{Ooi2002} and Klemel{\"a} \cite{Klemela2006} for extensions
and related ideas.
More recently,
more sophisticated multiscale methods which do not rely on variable
bandwidth kernel estimators have been developed for this purpose as
well, for example, D{\"u}mbgen and Spokoiny \cite
{DumbgenSpokoiny2001}, Davies and Kovac \cite{DaviesKovac2004}, D{\"
u}mbgen and Walther \cite{DumbgenWalther2008},
Schmidt-Hieber \textit{et al.} \cite{SHMD13}.

In this paper, we are concerned with circular data for which rigorous
inferential methods on the number and location of modes have not been
established yet to the best of our knowledge. Recently, Oliveira \textit{et al.} \cite
{OliveiraCrujeirasRodrigues-Casal2013} suggested a circular version of
the SiZer, and argued that this is of particular use for their problem
at hand, the analysis of Atlantic wind speeds and directions, however,
without providing a circular scale space theory or methods assessing
the statistical significance of empirically found modes. Also in the
application of the present paper, there are both biological and
practical reasons that let us also favor the relatively simple circular
SiZer's methodology although we are aware of some potential drawbacks,
for example, a loss in power and asymptotic accuracy for small scales.
In the below study of early differentiation of human mesenchymal stem
cells, the structure characteristic for a specific cell type appears to
be of the relative size captured by the first few
largest modes while modes on much smaller scales most likely feature
individual cell effects, not of immediate interest. The circular SiZer
will allow to investigate the coarser mode structure and relative
importance of modes in terms of bandwidth, which renders it a
relatively simple tool for visualization of dynamics of modes through
circular scales.

\textit{Causality.}
For data on the real line, it has been shown by Chaudhuri and Marron
\cite
{ChaudhuriMarron2000} that the SiZer controls the estimated modes at a
large range of scales (represented by a bandwidth~$h$) as long as these
are above a smallest scale $h_0>0$, say. Crucial for a valid
interpretation of these modes is ``causality'' which prevents the
estimator from creating artifacts with increased smoothing.
More specifically, causality of a family of kernels $\{L_h:h>0\}$ is
the requirement that for any integrable function $f$,
%
\begin{equation}
\label{var-dimi-intro:def} h\mapsto\#\ \mathrm {Modes}(L_h*f)\qquad\mbox{is decreasing.}
\end{equation}
Note that on the circle, the number of modes of a differentiable
function is simply half the number of sign changes of its derivatives
or less (cf. Definition~\ref{mode-reducing:def}).

On the real line, it is well known that the Gaussian kernel yields the
only causal family under suitable assumptions (Lindeberg \cite{Lindeberg1994}).
Nevertheless, violation of causality will usually only matter for
relatively small ``critical'' scales (Hall \textit{et al.} \cite
{hallminottezhang04}), and
other (compactly supported) kernels may be used without too much
concern as long as the true number of modes is small. However, for the
SiZer it appears to be difficult to decide for given data whether and
when such a critical scale is achieved, and hence causality violation
may be a concern.

\begin{figure}[b]

\includegraphics{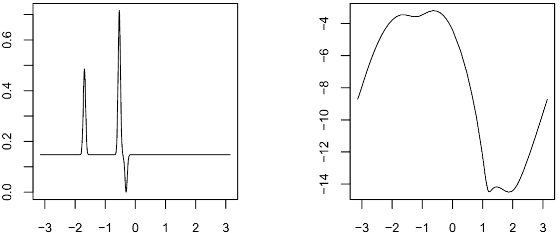}

\caption{The von Mises kernel is not causal (mode reducing)
on the circle; displaying on the right the tri-modal circular density
(its nonuniform part in log-scale) of a convolution with a von Mises
kernel ($\kappa=6$) of the bi-modal circular distribution displayed on
the left (its density is the combination of two spikes and a cleft with
a uniform distribution).}\label{von-Mises:fig}
\end{figure}

The above quoted circular SiZer method of Oliveira \textit{et al.} \cite
{OliveiraCrujeirasRodrigues-Casal2013} is a numerical extension of the
SiZer methodology from the line to the circle based on the von-Mises
kernel $g(e^{it}) \sim\exp(\kappa\cos(t - \mu))$ 
with variable concentration parameter $\kappa$ playing the role of the
(inverse) bandwidth. The von Mises density is often considered as the
natural counterpart of the Gaussian on the circle, for example,
maximizing entropy with given mean ({Mardia} and {Jupp} \cite{mardiajupp00},
Section~3.5.4).
For circular nonparametric smoothing, this kernel has also been used,
for example, in Fisher and Marron \cite{FisherMarron2001}, Taylor
\cite{Taylor2008}. However, Lemma~1 of
Munk \cite{Munk1999} states that the von Mises kernel is not variation
diminishing (i.e., the number of sign changes is not nonincreasing
under smoothing, cf. Definition~\ref{variation-diminishing:def}) for
$\kappa>1/2$ (the precise statement is that ``order 2 variation
diminishing''~-- which implies variation diminishing -- holds if and
only if $0\leq\kappa\leq1/2$). From our Theorem~\ref
{various-axioms:th}, it follows now that a convolution with a von Mises
distribution with concentration parameter $\kappa> 1/2$ violates
circular causality (\ref{var-dimi-intro:def}). Figure~\ref{von-Mises:fig} shows such an example where the convolution
with a von Mises kernel increases the number of modes (maxima) from two
to three.
As a consequence, the von Mises kernel seems questionable in the
context of the SiZer since a lower bound for the number of modes of the
smoothed density does not imply a lower bound for those of the true
underlying density. Therefore, the question remains whether a circular
causal kernel exists admitting a circular scale space approach.

\textit{Circular scale spaces -- The uniqueness of the wrapped Gaussian.}
In this paper, we give an affirmative answer to this question. Let
%
\begin{equation}
\label{wrapGauss:eq} K_{h}\bigl(e^{it}\bigr)=\frac{1}{\sqrt{h}}
\sum_{k\in\ZZ}\phi \biggl(\frac
{t+2\pi
k}{\sqrt{h}} \biggr)
\end{equation}
be the wrapped Gaussian kernel (on the circle) with bandwidth $h>0$
where $\phi$ denotes the standard normal density. We will show that the
wrapped Gaussian
is the only circular causal kernel. Therefore, from this perspective it
can be viewed as the most natural analog to the normal density on the
real line. This result requires some preparation.
We will:
\begin{longlist}[(1)]
\item[(1)] propose circular scale space axiomatics and discuss its
relationships to circular variation and mode reducing properties,
\item[(2)] show that under reasonable assumptions the wrapped Gaussian
kernel gives the one and only semigroup guaranteeing causality.
\end{longlist}
This will then allow us to:
\begin{longlist}[(3)]
\item[(3)] assess asymptotically the statistical significance of shape
features obtained from the WiZer,
\item[(4)] and to infer persistence over smoothing scales of shape features.
\end{longlist}

The analogs of (1) and (2) are well investigated in the linear case (cf.
Lindeberg \cite{Lindeberg1994}, Weickert \textit{et al.} \cite
{WeickertIshikawaAtuschi1999}). This result also
complements numerous linear scale-space axiomatics for images that have
attempted to establish uniqueness of the Gaussian kernel; see, for
example, Weickert \textit{et al.}~\cite
{WeickertIshikawaAtuschi1999} for an overview. However,
it should
be noted that a justification of Gaussian smoothing via a variation
diminishing axiom is only possible for 1-D signals (cf. Babaud
\textit{et~al.} \cite
{BWBD86}). It is well known that for 2-D images, Gaussian convolution
can create new extrema (see Lindeberg \cite{Lindeberg1994}). To avoid
this, one
has to consider scale-spaces based on nonlinear partial differential
equations such as mean curvature motion or the affine morphological
scale-space. For more details, we refer to Alvarez
\textit{et~al.} \cite{AGLM93}.
In contrast to the Euclidean case, circular scale spaces have
hardly been considered so far. Notable exceptions are a periodic
scale-space filtering by Wada \textit{et al.} \cite{WGS91} and a feature
extraction study by
Briggs \textit{et~al.} \cite{BDMS04}. However, there is no axiomatic
foundation of circular scale-spaces in terms of a variation
diminishing property. We will show in the following that on the circle
the situation is indeed peculiar and, for example, causality is fully
equivalent to the circular diminishing property. In summary, as a
byproduct of this paper, the circular case fills a gap of classical
scale space axiomatics.

\textit{Statistical guarantees.}
Concerning (3) above, we find for the SiZer based on wrapped Gaussians
that the probability of overestimating the true number of modes on
\emph
{any} scale can be bounded by some small number $\alpha$, say, while
the probability of underestimating the true number goes to zero for any
fixed bandwidth as the sample size goes to infinity.
Causality as in (\ref{var-dimi-intro:def}) ensures that, if we find
significant evidence for, say, $k$ circular modes in the smoothed
density, we automatically have significant evidence that there are also
at least $k$ modes in the true density.

\textit{Inferred mode persistence diagrams.}
Of particular importance are those scales where the lower bounds on
the number of inferred modes change and the intervals of scales in
which these bounds are constant. More general, under causality, the
latter give a notion of \emph{significant scale space persistences of
shape features} over smoothing scales. In fact, these expand the
general scheme of topological persistence introduced by Edelsbrunner \textit{et al.} \cite
{EdelsbrunnerLetscherZomorodian2002}, cf. also Ghrist \cite
{Ghrist2008}, Carlsson \cite{Carlsson2009} among many others. This originally
deterministic scheme is currently gaining high momentum in statistics
and medical imaging (e.g., Chung \textit{et al.} \cite{ChungBubenikKim2009},
Heo \textit{et al.} \cite{HeoGambleKim2012}).
Illustrated by the application at hand,
we propose \emph{inferred persistence diagrams} simultaneously
depicting significant bandwiths for
\emph{births} and \emph{splits} of modes. In addition to statistically
estimating (as done for persistent homologies by Bubenik and Kim \cite
{BubenikKim2007}), we also provide for confidence statements for the
estimate being a lower bound. This is related to recent work by
Balakrishnan
\textit{et~al.}
\cite
{BalakrishnanFasyLecciRinaldiAartiWasserman2013} who provide for
inference on persistence diagrams and by Schwartzman \textit{et al.} \cite
{SchwartzmanGavrilovAdler2011,SchwartzmanJaffeGavrilovMeyer2013} for
inferential mode detection of linear densities within specific classes.
Complementing this, however, we are firstly ``truly'' circular. Secondly,
we do not infer on the support of the data but on the shape of its very
density, and thirdly, for this density we require no assumptions.

\textit{A measure for early stem cell differentiation.}
Utilizing inferred persistence, we give a proof of concept to use
persistence diagrams of shape parameters -- namely modes -- to
elucidate precisely how the elasticity of the micro-environment directs
early human mesenchymal stem cell (hMSC)
differentiation;
hMSCs from bone marrow are considered highly promising for
regenerative medicine and tissue engineering. However, for successful
therapeutic applications understanding and controlling of
differentiation mechanisms is of paramount importance. The seminal
study of Engler
\textit{et~al.} \cite{EnglerSenSweeneyDisher2006} demonstrated that the
mechanical properties (Young's modulus $E$) of the micro-environment
can direct stem cell differentiation.
In Section~\ref{application:scn}, we re-analyze the fluorescence
microscopy images of nonmuscle myosin IIa in stem cells from Zemel
\textit{et~al.} \cite
{Zemel2010a}. Here, hMSCs have been cultured for 24~hours on substrates
of varying Young's modulus $E$ mimicking the physiological mechanical
properties of different tissue {in vivo}: 1~kPa corresponding to
neural cells and brain tissue, 11~kPa corresponding to muscle fibres
and 34~kPa corresponding to pre-mineralized collagen densifications in
bone, cf. Rehfeldt \textit{et~al.} \cite{Rehfeldt2007}. From these images, we extract
filament-process valued data\footnote{Which are provided for along with
the WiZer R-package
including inferred persistences
on the website of the second author,
\surl{http://rayerk.weebly.com/files.html}.} using the filament sensor of Eltzner
\textit{et~al.}
\cite
{EltznerWollnikGottschlichHuckemannRehfeldt2014}, in particular total
fibre length over angular directions, which live on the circle.
We do not only reproduce the observed nonmonotonicity over matrix
rigidity, now for mode persistences, but also, more subtly, the
respective mode-persistence diagrams indeed reflect different matrix
rigidities elucidating the impact of micro-environment rigidities on
early stem cell differentiation.

\textit{Plan of the paper.} In the following Section~\ref{scale-space:scn}, we begin by establishing circular scale space
axiomatics, while the details on circular sign changes have been
deferred to Appendix~\hyperref[appa]{A}. All of the proofs of the theorems subsequently
developed in this and in Section~\ref{inference:scn} have been deferred
to Appendix~\hyperref[appb]{B}. For the application to human stem cell differentiation
in Section~\ref{application:scn}, in Appendix~\hyperref[appc]{C} we give a numerical
foundation for the choice of the number of wrappings required to
compute the $p$-values underlying the tests performed by the WiZer with
a desired accuracy. For the bandwidths considered in our application,
an error less than $10^{-4}$ can be achieved with six wrappings.

\section{Circular scale space theory}\label{scale-space:scn}

\subsection{Notation}
Let $
\SS=\{z\in\CC:|z|=1\}
$ be the unit circle in the complex plane $\CC$
equipped with the measure $d\mu(z) = \frac{dt}{2\pi}$ for $z=e^{it}$,
with $t\in[0,2\pi)$ of the uniform distribution on $\SS$. This is the
normalized Haar measure on the compact Abelian group $\SS$ with the
ordinary multiplication of complex numbers. Then
\[
\hat f_k:= \int_\SS f(z)z^{-k} \,d
\mu(z)
\]
is the $k$th Fourier coefficient, $k\in\mathbb Z$, of a function $f\in
L^1(\SS)$.
Moreover, for $f,g\in L^1(\SS)$ we have the convolution
\[
(g*f) (w) = \int_\SS g\bigl(wz^{-1}\bigr)f(z) \,d
\mu(z) = \frac{1}{2\pi}\int_0^{2\pi}g
\bigl(e^{i(s-t)}\bigr) f\bigl(e^{it}\bigr) \,dt,\qquad w=e^{is}
\in\SS,
\]
which is well defined for $w\in\SS$ $\mu$-a.e. and in $L^1(\SS)$.

In particular, every function $g \in L^1(\SS)$ with $\int_\SS g(z)
\,d\mu
(z)=1$ is a \emph{circular kernel} as it conveys a bounded operator $g:
L^1(\SS) \to L^1(\SS), f\mapsto g*f$. A family of circular kernels $\{
L_h:h>0\}$ generates from every $f\in L^1(\SS)$ a \emph{scale space tube}
\[
\bigl\{L_h*f (z) : z\in\SS, h>0\bigr\}.
\]
We write $\hat L_{h,k}$ for the $k$th Fourier coefficients of $L_h$,
$k\in\mathbb Z$.

Finally, we say that a function $f:\SS\ra\RR$ is differentiable if the
\emph{counter-clockwise derivatives}
\[
(Df) (z )=\frac{d}{dt}\Big|_{t=0} \bigl(t\mapsto f
\bigl(e^{it}z \bigr) \bigr)
\]
exist for all $z\in\SS$, it is $m$-times differentiable if the $m$th
counter-clockwise derivatives exist which are defined iteratively as
$D^{m}f=D (D^{m-1}f )$.

\subsection{Regular scale space tube axiomatics}

In the following, $S_c(f)$ denotes the number of \emph{cyclic sign
changes} of a function $f:\SS\to\mathbb R$ as introduced by
Mairhuber \textit{et al.} \cite
{MairhuberSchoenbergWilliamson1959} and detailed in the \hyperref[app]{Appendix}.

\begin{Def}\label{variation-diminishing:def} A family of circular
kernels $\{L_h:h>0\}$ is
\begin{longlist}[(SG)]
\item[(SG)] a \emph{semigroup} if $L_h*L_{h'} = L_{h+h'}$ for all $h,h'>0$,
\item[(VD)] \emph{variation diminishing} if
\[
S_c(L_h*f)\leq S_c(f)
\]
for every $f\in L^1(\SS)$ and $h>0$,
\item[(SM)] \emph{symmetric} if $L_h(z) = L_h(z^{-1})$ for all $z\in
\SS
$ and $h>0$,
\item[(SL)] \emph{strongly Lipschitz} if there exists $r>0$ such that
the limit
\[
\lim_{h\ra0} \biggl(\frac{\hat{L}_{h,k}-1}{h|k|^r} \biggr)_{k\in\ZZ}
\]
exists in the space $\lw^\infty(\ZZ)$ of bounded sequences on $\ZZ$.
\end{longlist}
A scale space tube generated by a strongly Lipschitz, symmetric and
variation diminishing semigroup will be called a \emph{regular scale
space tube}.
\end{Def}

Straightforward computation gives for the wrapped Gaussian kernel that
$\hat K_{h,k} = e^{-k^2h/2}$, and hence the following.

\begin{Rm}\label{gauss-is-regular:rm}
The family of wrapped Gaussian kernels is a symmetric semigroup that is
strongly Lipschitz with $r=2$.
\end{Rm}

Of special interest are those families $\{L_h: h>0\}$ of kernels that
are \emph{differentiable} in the sense that
\[
\partial_h (L_h*f) = \lim_{t\to0}
\frac{L_{h+t}*f-L_h*f}{t}
\]
exists for all $h>0$ and $f\in C^\infty(\SS)$. The following is a
straightforward consequence of the definition.

\begin{Rm} \label{rem:diffSemiGroup}A semigroup of circular, strongly
Lipschitz kernels is differentiable. To see this, consider the Fourier
coefficients
\[
\hat{ \bigl( (L_{h+t}*f-L_h*f )/t \bigr)_k}
= \hat{L}_{h,k} (\hat{L}_{t,k}\hat{f}_k-
\hat{f}_k )/t = \hat{L}_{h,k} |k|^r
\hat{f}_k (\hat{L}_{t,k}-1 )/\bigl(t|k|^r\bigr).
\]
The latter has a limit in $\lw^2(\ZZ)$ for $t\ra0$ since smoothness of
$f$ guarantees that $ (\hat{L}_{h,k} |k|^r\hat{f}_k )_k$
is in
$\lw^2(\ZZ)$ and $\lim_{h\ra0} (\hat{L}_{t,k}-1 )/(t|k|^r)$
exists in $\lw^\infty(\ZZ)$.
\end{Rm}

At this point let us relate the heat equation
%
\begin{equation}
\label{heat:eq}\tfrac{1}{2}\Delta u =\partial_h u
\end{equation}
on a Riemannian manifold $M$ with Laplace--Beltrami operator $\Delta$
for functions $u : M \times[0,\infty)\rightarrow\RR$ with some
initial condition $u(\cdot,0)=f$ (e.g., Sakai \cite{Sakai1996}, Section IV.3)
to scale space axiomatics. Using the Fourier transform, the heat
equation can be solved in the spectral domain such that the unique
solution of (\ref{heat:eq}) is given by the convolution of the initial
condition with the usual Gaussian $x\mapsto\phi(x/\sqrt{h})/\sqrt{h}$
(in case of $M=\RR$) and the wrapped Gaussian $z\mapsto K_h(z)$ (in
case of $M=\SS$), if the convolution exists; where in the circular case
$\Delta= D^2$.
Requiring that equality~(\ref{heat:eq}) holds only for the signs of
both sides
at proper local maxima or minima, that is, that maxima and minima are
not enhanced
under smoothing, is another axiom frequently found in the scale space
literature. This axiom has been used by Babaud
\textit{et~al.} \cite{BWBD86},
for deriving the Gaussian as the unique scale space kernel in 1D.
Later on, Lindeberg \cite{Lindeberg2011} emphasised that
nonenhancement of
local extrema distinguishes the Gaussian from other scale space kernels
such as the ones corresponding to the so-called $\alpha$ scale spaces.
However, if one goes beyond scale space representations that can be
expressed in terms of convolutions with kernels, this property can also
be fulfilled. For instance, nonenhancement of local extrema has been
established in Weickert \cite{Weickert1998}, Section~2.4.2, for a
family of nonlinear
diffusion evolutions. Similar reasonings can be used in the circular
setting as well.

\begin{Def}
A family of differentiable circular kernels $\{L_h:h>0\}$ is
\begin{longlist}[(NE)]
\item[(NE)] \emph{not enhancing local extrema} if for every smooth
$f:\SS\to\RR$ and $h>0$
\[
\sign \bigl(D^2 (L_h* f) \bigr) \sign \bigl(
\partial_h (L_h* f) \bigr)\geq0
\]
at every nondegenerate critical point (i.e., points at which $D (L_h*
f)=0\neq D^2 (L_h* f)$) of $L_h* f$ on $\SS$.
\end{longlist}
\end{Def}

Finally, we formalize (\ref{var-dimi-intro:def}). From this, we obtain
one more axiom.

\begin{Def}\label{mode-reducing:def}
Let $k\in\mathbb N_0$. A differentiable function $f\in L^1(\SS)$ is
\emph{$k$-modal} (i.e., it has $k$ modes or less) if
\[
S_{c}(Df) \leq2k.
\]
A circular kernel $L_h$, $h>0$, is
\begin{longlist}[(MR)]
\item[(MR)] \emph{mode reducing} (i.e., reducing the number of modes)
if for any function $f$ that is differentiable except for finitely many
points and that satisfies $ \lim_{t\downarrow0}Df(e^{it}z)=\lim_{t\uparrow0}Df(e^{it}z)$ for all $z\in\SS$, we have
\[
S_c (L_h*Df ) \leq S_{c}(Df),
\]
where $Df$ is continued to all of $\SS$ via its left and right limits.
\end{longlist}
\end{Def}

\begin{Th}\label{various-axioms:th}
A family of circular kernels is variation diminishing if and only if
it is mode reducing. Moreover, a differentiable semigroup of variation
diminishing circular kernels is not enhancing local extrema.
\end{Th}

Here is the complete picture.
\begin{eqnarray*}
\lleft. %
\begin{array} {rcccl} &&\mathrm{(MR)} &\Leftrightarrow& \mathrm{(VD)}
\\
&&&&\mbox{and}
\\
&&\mathrm{(SL)}, \mathrm{(SG)}&\quad\Rightarrow\quad&\mbox{differentiable } \mathrm{(SG)} \end{array} %
 \rright\}&\Rightarrow& \textup{(NE)}.
\end{eqnarray*}

On the line, an analog definition of mode reducing is more simple, as
no periodicity of the anti-derivative is required. For functions
$f,g\in L^1(\RR)$, denote by ($f*_\RR g)(y) = \int_{-\infty}^\infty
f(x) g(y-x) \,dx$ the usual convolution on the line, if existent and by
$S^-(f)$ the number of sign changes of $f$. (cf. Karlin \cite
{Karlin1968}, Brown \textit{et al.} \cite
{BrownJohnstoneMacGibbon1981} and Appendix~\hyperref[appa]{A}.) Then call a
kernel $K$ on the line \emph{mode reducing} if for any differentiable
function $f:\RR\ra\RR$ we have $S^{-} ( (K*_\RR f
)'
)\leq S^{-}(f')$.

\begin{Rm} Inspection of the proofs of Lemmata \ref{MR_equiv_VR:lem}
and \ref{CV-implies-NE:lem} in the \hyperref[app]{Appendix} shows that on the line
$\RR
$, with the analog axioms, Theorem~\ref{various-axioms:th} also holds true.
\end{Rm}

\subsection{The uniqueness of the wrapped Gaussian kernel}
The analogs of the following two characterizations of the wrapped
Gaussian kernel as generating regular scale space surfaces and being up
to scaling the only variation diminishing, symmetric and strongly
Lipschitz semigroup are well known on the line and, under suitable
adaptations on Euclidean spaces. For an overview over various scale
space axiomatics for Euclidean spaces, see Weickert \textit{et al.}
\cite
{WeickertIshikawaAtuschi1999}. Suitably adapted arguments, sometimes
even simpler lead to the following two circular versions, which, to the
knowledge of the authors, have not been established before.

\begin{Th}\label{wrappedGauss-is-causal:thmm} The family of wrapped
Gaussian kernels $\{K_h: h>0\}$ generates a regular scale space tube.
\end{Th}

\begin{Th}\label{causal-is-wrappedGauss:thmm} Let $\{ L_{h}:h>0\}$ be a
semigroup generating a regular scale space tube.
Then there is a constant $\alpha>0$ such that for all $h>0$ the
function $L_{\alpha h}$ is the wrapped Gaussian kernel $K_h$ with bandwidth
$h$.
\end{Th}

\begin{Rm}\label{strong-continuity:rm}
Often in scale space literature, the axiom of \emph{strong continuity}
\begin{longlist}[(SC)]
\item[(SC)] $\|(L_h - {\rm I})f\|_{L^2} \to0$ as $h\downarrow0$ for
all $f\in L^2(\SS)$,
\end{longlist}
is introduced which is weaker than the strong Lipschitz property (SL).
Here, I denotes the identity. It seems that this property is not
sufficient to ensure the critical fact (detailed in the \hyperref[app]{Appendix}) that
the smooth functions are in the domain of definition of the
infinitesimal operator
\[
\mathcal{A} = \lim_{h\downarrow0} \frac{L_h-{\rm I}}{h} ,
\]
which is required if $\mathcal{A}$ is to turn out to be a multiple of
the Laplacian.

More precisely, we first show that $ \mathcal{A}$ is a local operator,
that is, if a smooth function vanishes in an neighborhood of a point,
then its image under $\mathcal{A}$ also vanishes at this point. This
implies that $\mathcal{A}$ is a differential operator. In a second
step, we argue that, due to causality, $\mathcal{A}$ must be a multiple
of the Laplacian.

This strain of arguments has been employed by Lindeberg \cite{Lindeberg2011},
page~41, for the Euclidean case, who has suggested to require
a property different from the strong Lipschitz property (SL) which we
require for the circle.
\end{Rm}

\section{Inference on shape parameters}\label{inference:scn}

In the previous Section~\ref{scale-space:scn}, we showed that the
wrapped Gaussian generates a regular scale space tube and that it is
only the wrapped Gaussian that has this property. For this reason, in
the following, we only consider the wrapped Gaussian kernel $K_h$
defined in (\ref{wrapGauss:eq}) although empirical scale space tubes
(Definition~\ref{scale_space_tube:def}) can of course be defined for
arbitrary families, Remark~\ref{Rem:Info_small_bandwidths} holds for
any semigroup and the proofs of Theorem~\ref{thmm:Weak_Conv} and
Corollary~\ref{weak-conv:cor} require only that second moments be
finite. In the literature, the properties variation diminishing,
nonenhancement of modes, etc. are often referred to as preserving
\emph
{causality}. Thus, the wrapped Gaussians are the one and only (under
reasonable assumptions) kernels preserving circular causality.

\subsection{Circular causality}\label{causality:scn}

\begin{assumption}
\label{Assumption}From now on assume that $X$ is a random variable
(with or without a density with respect to the Haar measure
$\mu$) taking values on the circle $\SS$ and we observe $X_{1},\ldots
,X_{n}\stackrel{i.i.d.}{\sim}X$.
\end{assumption}

\begin{defn}\label{scale_space_tube:def}
We call the two-parameter stochastic process indexed in $\SS\times\RR^+$
\[
\Biggl\{ f_{h}^{(n)}(z):=\frac{1}{n}\sum
_{j=1}^{n}K_{h} \bigl(zX_{j}^{-1}
\bigr):z\in\SS,h>0 \Biggr\}
\]
the \emph{empirical circular scale space tube} and
\[
\bigl\{ f_{h}(z):=Ef_{h}^{(n)}(z):z\in\SS,h>0
\bigr\}
\]
the \emph{population circular scale space tube}.
\end{defn}

\begin{rem}\label{expectation:rm}
Note that if $X$ has density $f$ with respect to the Haar measure $\mu$
on $\SS$ then under Assumption~\ref{Assumption}
\[
E\bigl(f_{h}^{(n)}(z)\bigr)  = E \bigl(K_{h}
\bigl(zX^{-1} \bigr) \bigr) =\int_{\SS}K_{h}
\bigl(zw^{-1}\bigr)f(w)\mu(dw) = (K_{h}*f ) (z).
\]
\end{rem}

Now we can state the causality theorem.
%
\begin{thmm}[(The circular causality theorem)]\label{thmm:circ-causality}
Let $K_h$ be the wrapped Gaussian. Under Assumption~\ref{Assumption},
for any $m=0,1,\ldots$ the following holds:
\begin{longlist}[(ii)]
\item[(i)] The mappings $h\mapsto S_{c} (D^{m}f_{h} )$ and
$h\mapsto S_{c} (D^{m}f_{h}^{(n)} )$
are decreasing and right continuous functions on $(0,\infty)$.
\item[(ii)] If $X$ has a density $f$ with respect to $\mu$ which is $m$-times
differentiable and we set $f_{0}=f$ then $h\mapsto S_{c}
(D^{m}f_{h} )$
is decreasing and right continuous on $[0,\infty)$.
\end{longlist}
\end{thmm}
%
\subsection{Weak convergence of the scale space tube}\label
{weak-convergence:scn}

\begin{rem} \label{Rem:Info_small_bandwidths}All information of the empirical
and the population scale space tube for any bandwidth is already
contained in that of every smaller bandwidth. More precisely, for
$h_0>0$, $\{f_{h}^{(n)}(z):z\in\SS,h\geq h_{0}\}$
and $\{f_{h}(z):z\in\SS,h\geq h_{0}\}$ can be reconstructed from
$\{f_{h_{0}}^{(n)}(z):z\in\SS\}$ or $\{f_{h_{0}}(z):z\in\SS\}$,
respectively.
\end{rem}

The reason for this is of course the identity $f_{h_{0}+h}=K_{h}*f_{h_{0}}$
for any $h>0$ (and similarly for the empirical counterpart), which also
holds for $h=0$ if we set $K_0*f := f$.

Moreover, $(z,h)\mapsto K_{h}(z)$ is a solution to the heat equation
(\ref{heat:eq}), and hence $f_{h}^{(n)}(z)$ and $f_{h}(z)$ are
solutions as well.
By the well known maximum principle for solutions of the heat equation,
maxima are attained at the boundary, that is, we have
%
\begin{equation}
\label{max-principle-PDE:eq} \sup_{z\in\SS,h\geq h_{0}}
\bigl|f_{h}^{(n)}(z)-f_{h}(z)\bigr|=
\sup_{z\in
\SS
}\bigl|f_{h_{0}}^{(n)}(z)-f_{h_{0}}(z)\bigr|\qquad
\mbox{for all }h_0 > 0.
\end{equation}

\begin{thmm}
\label{thmm:Weak_Conv}Under Assumption~\ref{Assumption}, let $m\geq0$
be an integer and define
%
\begin{equation}
\cv(z_{1},z_{2};h)= \operatorname{cov}
\bigl(D^{m}K_{h}\bigl(z_{1}X^{-1}
\bigr),D^{m}K_{h}\bigl(z_{2}X^{-1}
\bigr)\bigr)\label{eq:cov_def}
\end{equation}
for $z_{1},z_{2}\in\SS$ and $h>0$. Then, for any fixed $h>0$,
%
\begin{equation}
n^{{1}/{2}} \bigl(D^{m}f_{h}^{(n)}(z)-D^{m}f_{h}(z)
\bigr)\ra G_{h},\label{eq:G_h0}
\end{equation}
weakly in $C(\SS)$, where $G_{h}$ is a Gaussian process on $\SS$ with
mean zero and covariance
structure defined by $\cv(z_{1},z_{2};h)$.

Moreover, $G_{h}$ has continuous sample paths with probability
one. In particular, $P(\sup_{z\in\SS}|G_{h}(z)|<\infty)=1$.
\end{thmm}

This gives at once the following.

\begin{cor}\label{weak-conv:cor}
Under Assumption~\ref{Assumption} with $h_0> 0$,
\[
\sup_{z\in\SS,h\geq h_0}n^{
{1}/{2}}\bigl|D^{m}f_{h}^{(n)}(z)-D^{m}f_{h}(z)\bigr|
\]
converges weakly to $\sup_{z\in\SS}|G_{h_0}(z)|$ as $n\ra\infty$,
where $G_{h_0}$ is a Gaussian process with zero mean and covariance
structure $\cv(z_{1},z_{2},h_0)$ given by (\ref{eq:cov_def}).
\end{cor}

In the following denote by $q_{(1-\alpha)}$ the $(1-\alpha)$ quantile
of the random variable $\sup_{z\in\SS}|G_{h_0}(z)|$, for $0\leq
\alpha
\leq1$.

The convergence results can be used to test whether the derivative
$D^{m}f_{h}(z)$ at some point $z\in\SS$ and for some bandwidth $h\geq h_{0}>0$
is different from zero. We perform the tests for the hypotheses
\[
H_{0}^{(h,z)}:D^{m}f_{h}(z)=0,\qquad  z\in\SS,h
\geq h_{0}
\]
as follows:
%
\begin{equation}
\mbox{If } %
\cases{ D^{m}f_{h}^{(n)}(z)>n^{-{1}/{2}}q_{(1-\alpha)},
&\quad $\mbox{reject }H_{0}^{(h,z)}\mbox{ and conclude
}D^{m}f_{h}^{(n)}(z)>0$,\vspace *{2pt}
\cr
\bigl|D^{m}f_{h}^{(n)}(z)\bigr|\leq n^{-{1}/{2}}q_{(1-\alpha)},
&\quad $\mbox {accept }H_{0}^{(h,z)}$,\vspace*{2pt}
\cr
D^{m}f_{h}^{(n)}(z)<-n^{-{1}/{2}}q_{(1-\alpha)},
&\quad $\mbox{reject }H_{0}^{(h,z)}\mbox{ and conclude
}D^{m}f_{h}^{(n)}(z)<0$.} %
\label{eq:Test}
\end{equation}
The following theorem states two key properties of this test:

Firstly, by using the quantile $q_{(1-\alpha)}$ of the supremum $\sup_{z\in\SS}|G_{h_{0}}|$
we control the \emph{family-wise error rate} (FWER) of the tests.
More precisely, we can assert that the asymptotic probability that
one or more of the hypotheses $\{H_{0}^{(h,z)}:z\in\SS,h\geq h_{0}\}$
is falsely rejected is at most $\alpha$.

Secondly, we see that for a given bandwidth $h$ all sign changes
of $D^{m}f_{h}$ are detected by this test asymptotically with probability
one. Note that this does not require the test to detect all points
of positive/negative derivative. Indeed, without prior information on
the smallest scales this is not possible for any test since the
absolute value of the derivative can be arbitrarily
small.

\begin{thmm}\label{thmm:TestThm}
If under Assumption~\ref{Assumption} with $h_0>0$, $0\leq\alpha\leq
1$, the test (\ref{eq:Test})
is performed for each of the hypotheses $\{H_{0}^{(h,z)}:z\in\SS
,h\geq
h_{0}\}$
then the probability that one or more of them are falsely rejected
is asymptotically at most $\alpha$.

Moreover, if for any fixed bandwidth $h\geq h_{0}$ the function $D^{m}f_{h}(x)$
has $2k\geq1$ sign changes then this test will detect them with
asymptotic probability one.
\end{thmm}

\subsection{Inferred persistence}\label{persistence:scn}
A scale space satisfying causality gives a notion of \textit
{persistence} of features of a (density-) function $f$ which are given
by zero-crossings of derivatives. The number of such features is a
decreasing integer valued function in the bandwidth and, therefore, is
constant except for finitely many jumps. We can consider the bandwidths
associated with these jumps as the amount of smoothing that is
necessary to remove the features, one by one. Indeed, with the wrapped
Gaussian $K_h$ as kernel, we can define a sequence decreasing in $k$,
\[
\ssp^{(m)}_k := \ssp^{(m)}_k(f) = \inf
_{h>0} \bigl\{h:S_c\bigl(K_h *
D^mf\bigr) < 2k\bigr\} ,\label{eq:ssp}
\]
and call $\ssp^{(m)}_k$ the \textit{scale space persistence} of $D^{m-1}f$.

We will now use the family of tests introduced in (\ref{eq:Test}) to
define an empirical counterpart of $\ssp^{(m)}$, which can be obtained
from the data $X_1,\ldots,X_n$.
To this end, let
%
\begin{eqnarray}\label{eq:signatureFcn}
W^{(m)}_{h}(z)= %
\cases{ +1, & \quad$\mbox{if
}D^mf_{h}^{(n)}(z)>n^{-{1}/{2}}q_{(1-\alpha
)}$,
\vspace*{2pt}
\cr
0, &\quad $\mbox{if }\bigl|D^mf_{h}^{(n)}(z)\bigr|
\leq n^{-
{1}/{2}}q_{(1-\alpha
)}$,\vspace*{2pt}
\cr
-1, &\quad $\mbox{if
}D^mf_{h}^{(n)}(z)<n^{-{1}/{2}}q_{(1-\alpha)}$
}
\nonumber
\\[-6pt]
\\[-8pt]
\eqntext{\mbox{ for }h\geq h_0\mbox{ and } W_h =
W_{h_0}\mbox{ for }h<h_0,\qquad}
\end{eqnarray}
which we call the \emph{WiZer signature function},
and define
%
\begin{equation}
\widehat{\ssp}^{(m)}_k := \widehat{\ssp}^{(m)}_k(f)
= \inf_{h>0} \bigl\{ h:S_c\bigl(W^{(m)}_h
\bigr) < 2k\bigr\}\label{eq:emp_ssp}
\end{equation}
as the \emph{significant empirical scale space persistence} of
$D^mf_{h}$ or just \emph{inferred persistence}. Note that we are \emph
{not} simply defining the empirical scale space persistence as the
persistence of the kernel density estimator $\inf_{h>0} \{h:S_c(D^m
f^{(n)}_h) < 2k\}$. The reason is, that we want to eliminate
statistically insignificant features.

The following is an immediate consequence of Theorem~\ref{thmm:TestThm}.

\begin{cor}
Under the assumptions of Theorem~\ref{thmm:TestThm}, the following holds
for the scale space persistence of $D^{m-1}f$
\[
\liminf_{n\ra\infty}P \bigl(\widehat{\ssp}_k^{(m)}
\leq\ssp _k^{(m)}\mbox{ for all }k\in\ZZ \bigr)\geq1-\alpha.
\]
\end{cor}

\begin{rem}
Note that by definition the sequence $\ssp^{(m)}_1,\ssp^{(m)}_2,\ssp
^{(m)}_3,\ldots$
is decreasing. For $k\geq1$, $\ssp^{(m)}_k$ can be considered as the
\emph{birth} bandwidth of the $k$th mode of $D^{m-1} f_h$ and for
$k\geq2$ it is also the \emph{split} bandwidth where $k-1$ modes of
$D^{m-1} f_h$ \emph{split} into $k$ modes. As birth and splits (or
births) occur with decreasing bandwidths, our scenario is twofold
opposite to that of usual persistence diagrams (e.g., Edelsbrunner \textit{et al.} \cite
{EdelsbrunnerLetscherZomorodian2002}, Cohen-Steiner \textit{et al.} \cite{CohenEdelsbrunnerHarer2007},
Chung \textit{et al.} \cite{ChungBubenikKim2009}) where births and
mergers (or deaths) occur with
increasing bandwidth.
\end{rem}

The application in the following Section~\ref{application:scn}
illustrates the case $m=1$ of scale space persistence of modes of $f_h$.

%
\includegraphics{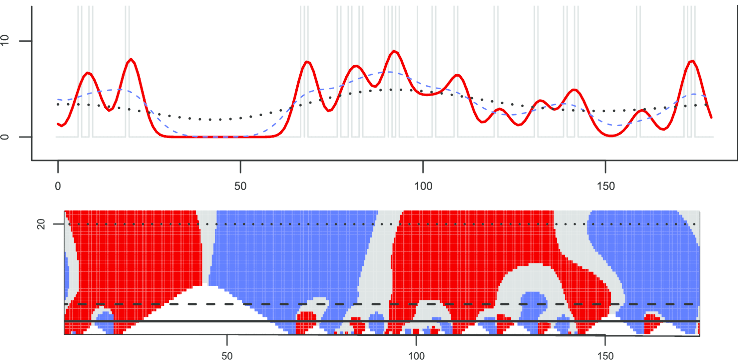}


\section{Application to early stem cell differentiation}\label
{application:scn}

\textit{Datasets -- acto-myosin cytoskeleton.}
We re-analyze a data set from Zemel \textit{et~al.} \cite{Zemel2010a}
hMSCs (human mesenchymal stem cells) cultured for 24~hrs. on elastic
substrates of different Young's moduli $E$ (1, 11 and 34~kPa),
subsequently chemically fixed and immuno-stained for NMM IIa, the motor
proteins making up small filaments that are actually responsible for
cytoskeletal tension.
Fluorescence images were recorded
for 30 cells on each of the three conditions:
\begin{longlist}[1.]
\item[1.]$E=1$~kPa which form the data set \emph{Black},
\item[2.]$E=11$~kPa which form the data set \emph{Red} and
\item[3.]$E= 34$~kPa which form the data set \emph{Blue}.
\end{longlist}

\begin{figure}

\includegraphics{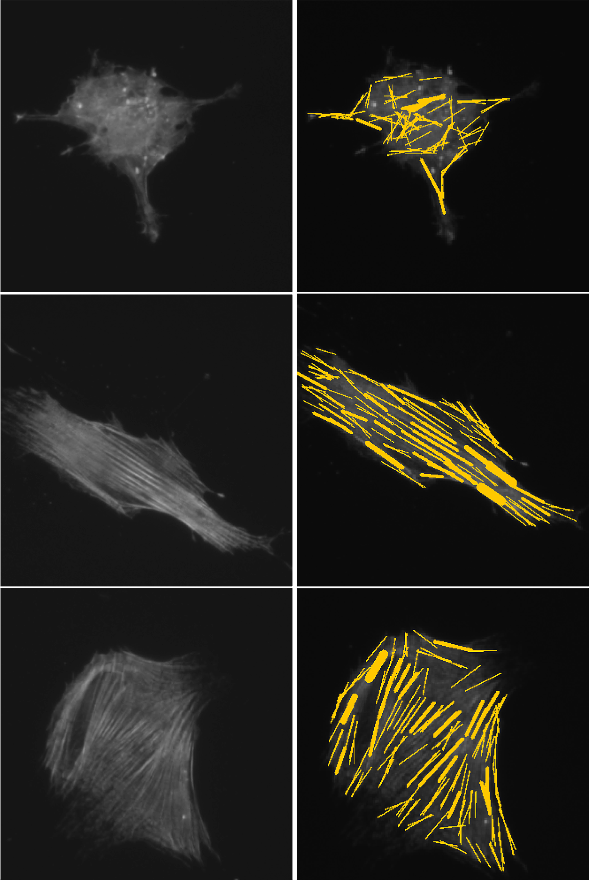}

\caption{Fluorescence microscopy images of human mesenchymal stem
cells, showing immuno-labeled myosin
structures, cultured for 24 hrs. on matrices of varying young moduli
$E_m = 1$~kPa (top row), $E_m = 11$~kPa (middle row) and $E_m = 34$~kPa
(bottom row). Left column: raw fluorescence images, right column:
superimposed traced filament structure.
(Colors are visible in the online version of
the article.)}\label{raw-traced:fig}
\end{figure}

From these micrographs, the filament structure is reconstructed using
the filament sensor of Eltzner
\textit{et~al.} \cite{EltznerWollnikGottschlichHuckemannRehfeldt2014}.
Figure~\ref{raw-traced:fig} shows typical cell images with their
traced skeleton from each of the three datasets.

\textit{The statistical model.} Each traced image gives a realization
$(z_j,w_j)_{j=1}^J$ of a bounded filament process where $J\in\mathbb
N$ denotes the number of filaments, $z_j\in\mathbb C$ the $j$th
filament's center in complex notation, $w_j = \lambda_je^{i\phi_j}$
encodes its length $\lambda_j>0$ and orientation angle $\phi_j \in
[-\pi
/2,\pi/2)$ with the positive real axis ($j=1,\ldots,J$). For a given
binning number $N\in\mathbb N$ we have the empirical histogram
%
\begin{equation}
\label{histogram:def} \Biggl(\sum_{j: \phi_j \in
[-{\pi}/{2} + {\pi k}/{N},-{\pi}/{2}+{\pi
(k+1)}/{N} )}
\lambda_j \bigg/ \sum_{j=1}^J
\lambda_j \Biggr)_{k=0}^{N-1}.
\end{equation}
We assume that there is a true underlying filament process such that
\[
\mathbb P\{z\in A,\lambda\in B,\phi\in C\} = \int_{A\times B \times C} g(z,
\lambda,\phi) \,dz \,d\lambda \,d\phi
\]
for Borel sets $A\subset\mathbb C, B\subset(0,\infty)$ and $C\subset
[-\pi/2,\pi/2)$ with a density $g$ w.r.t. the corresponding Lebesgue
measures. Then the histogram (\ref{histogram:def}) is an estimator for
the true conditional density
%
\begin{equation}
\label{marginal:eq}f(\phi) =\frac{\mathbb
E[\lambda
|\phi]}{\mathbb E[\lambda]},\qquad \phi\in[-\pi/2,\pi/2) .
\end{equation}

This statistical model relates to the previous theoretical analysis as
follows. For every observation of the filament process, there are $n$
pixels carrying orientations and we denote by $X_l \in[-\pi/2,\pi/2)$
the orientation of the $l$th pixel, $1\leq l\leq n$, where a pixel is
multiple counted, each with the orientation of the corresponding
filament, if two or more filaments intersect at this pixel. The binned
histogram of $X_1,\ldots,X_n$ is given by (\ref{histogram:def}).
Similarly, the filament process is a point process carrying weighted
(by relative filament length) orientation marks, where the distribution
of weighted orientations $X\in[-\pi/2,\pi/2)$ has the density $\phi
\mapsto f(\phi)$ given by (\ref{marginal:eq}). Taking into account
pixel and machine discretization, $X_1,\ldots,X_n$ can be viewed as a
discretized sample of the weighted orientation mark $X$. If we double
the orientations, we are in the situation of Assumption~\ref{Assumption}.

%

\textit{Empirical histograms and WiZer functions.}
From the empirical histograms with total number $n$ of observed pixels
(filament crossings result in pixels occurring in more than one
filament; such pixels are accordingly multiply counted), as well as
from the true density $f$ we obtain densities $f_h^{(n)}$ and $f_h$,
respectively, smoothed with the wrapped Gaussian kernel (\ref
{wrapGauss:eq}) of bandwidth $h>0$. For the former, the WiZer signature
functions $W_h$ are computed as in (\ref{eq:signatureFcn}). For the
hMSCs from Figures~\ref{raw-traced:fig}, \ref
{histogram-sizer.fig} depicts empirical histograms, corresponding
$f_h^{(n)}$ for some values of $h$ and the WiZer signature functions.

Note that not all of the modes of the smoothed empirical densities may
be statistically significant at the specified level, cf. Figure~\ref{Detail-histogram-sizer.fig}.

\begin{figure}

\includegraphics{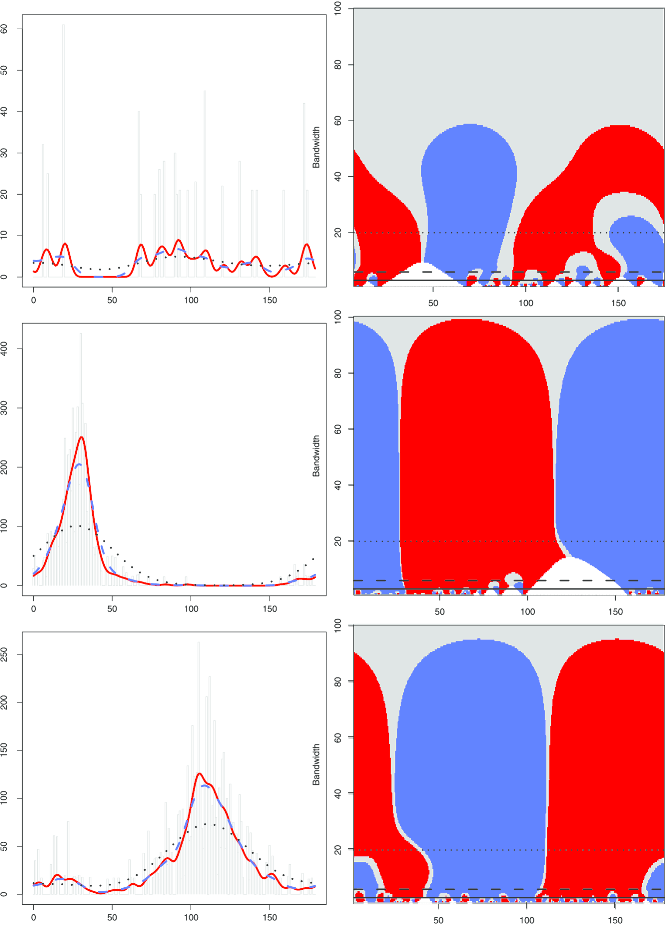}
\vspace*{-5pt}
\caption{Processed traced filament structures from Figure~\protect\ref{raw-traced:fig} (rows correspond). Left column: histograms of total
filament length $\times$ width (vertical) per angle orientation
(horizontal) with smoothed densities via wrapped Gaussian kernels of
bandwidths 3 (solid red), 6 (dashed blue) and 20 (dotted black). Right
column: the respective WiZer signature functions to the significance
level $\alpha= 0.05$ (horizontal: angle orientation, vertical:
bandwidth). Here, blue depicts significant increase $(W^{(1)}_h(z)=1)$,
red significant decrease\vspace*{-2pt} $(W^{(1)}_h(z)=-1)$, grey regions where
neither increase nor decrease is significant $(W^{(1)}_h(z)=0)$ and
white areas with too few data (effective sample size $\sum_{i = 1}^n
K_h (x - x_i) / K_h(0) \leq5$, Chaudhuri and Marron \cite{ChaudhuriMarron1999},
page~812).\vspace*{-1pt}
The three bandwidths from the left column are horizontal lines in the right.
(Colors are visible in the online version of
the article.)}
\label{histogram-sizer.fig}
\end{figure}

\textit{Inferred mode-persistence diagrams.} In the next step, the
inferred $k$th persistence bandwidths $\widehat{\ssp}^{(1)}_k$ for
modes of $f_h^{(n)}$ have been computed from the WiZer signature
functions. Recall that $\widehat{\ssp}^{(1)}_k$ gives the inferred
bandwidth of birth of the $k$th mode which is for $k>1$ the split
bandwidth of the $(k-1)$st mode. For the three hMSCs of the previous
Figures~\ref{histogram-sizer.fig} and \ref{Detail-histogram-sizer.fig},
Table~\ref{persistence:tab} depicts the first few inferred mode birth
bandwidths in degrees.

\begin{figure}

\includegraphics{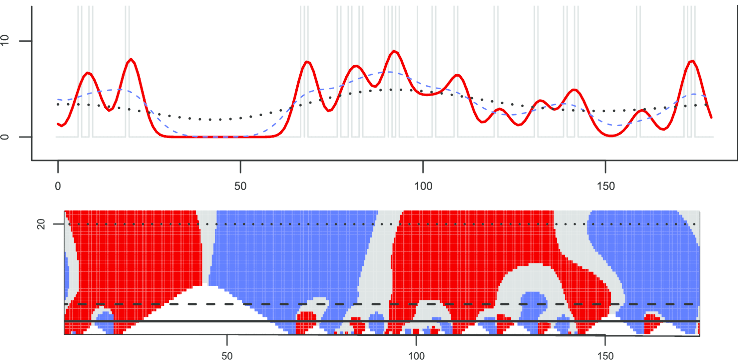}

\caption{Detail of the first row of Figure \protect\ref
{histogram-sizer.fig}. Careful investigation shows that only one of the
two first modes of the kernel smoothed empirical density with bandwidth
3 (red, solid) is statistically significant. Similarly, the first mode
of the kernel smoothed empirical density with bandwidth 6 (blue,
dashed) is statistically insignificant.
(Colors are visible in the online version of
the article.)}\label{Detail-histogram-sizer.fig}
\end{figure}

Finally, we depict inferred mode persistence bandwidths graphically.
Typically, in logarithmic scale, inferred persistences of several modes
can be included in a single diagram such that the verticals measure
inferred persistence of odd order modes while the horizontals measure
inferred persistence of even order modes, cf. Figure~\ref{persistence:fig}. With confidence level $1-\alpha= 0.95$, the black
cell has least persistent first modes and most persistent second and
higher modes. This is reversed for the red cell: most persistent first
mode and least persistent second and third mode. The blue cell's mode
persistences are intermediate. For all cells, the fourth mode comes to
lie on the first diagonal which accounts for the fact that fourth modes
have not been detected for any cell at the given resolution level (200
steps between the bandwidth reported in the WiZer signature function),
cf. Figure~\ref{persistence:tab}.

\textit{Data analysis.} We have now applied the above analysis to
the data set of a total of 179 cells: 60, 58 and 61 on each elasticity
(1~kPa, black, 11~kPa, red, 34~kPa, blue). Each of the three datasets
comes in two files reflecting two experimental batches. For each of
these six files, common parameters (reflecting average intensity, blur
and noise of the specific experimental setup) for the filament sensor
(cf. Eltzner
\textit{et~al.} \cite{EltznerWollnikGottschlichHuckemannRehfeldt2014}) have been
determined by an expert. Figure~\ref{persistence-sample-boxplots:fig}
depicts the boxplots for the first 9 inferred mode persistences and
Figure~\ref{persistence-sample-means:fig} the persistence diagram for
the means of the first seven, all in logarithms of radians. Recalling
that in the persistence diagram, odd order modes are more persistent if
they have a higher vertical component, even order are so if their
horizontal component is higher, the trend observed in Figure~\ref{persistence:fig} is
consolidated. In mean and median,

\begin{table}[b]
\caption{Inferred bandwidths $\widehat{\ssp}^{(1)}_k$ for
$k=1,\ldots,6$ and $k=16$ in degrees of births of modes with confidence
$1-\alpha= 0.95$ for the three cells from Figure~\protect\ref{raw-traced:fig}.
The first four are depicted in log-scale in
Figure~\protect\ref{persistence:fig}}
\label{persistence:tab}
\begin{tabular*}{\textwidth}{@{\extracolsep{\fill}}llllllllll@{}}
\hline
\mbox{Mode no.}&1&2&3&4&5&6&$\cdots$ &16&
$\cdots$
\\
\hline
\mbox{Cells on 1~kPa (black)}&58.420 &26.245 &8.425 &4.960 &4.960 &4.465&
&0.573
\\
\mbox{Cells on 11~kPa (red)}& 99.010 &\phantom{0}5.455 & 3.970 & 2.485 & 2.485 &
2.485& &1.495
\\
\mbox{Cells on 34~kPa (blue)} &95.050 &15.355 &4.960 &3.970 &3.970 &3.475&
&1.990 \\
\hline
\end{tabular*}
\end{table}

\begin{longlist}[1.]
\item[1.] cells on 1~kPa (black) have least persistent first modes and most
persistent next higher modes,
\item[2.] cells on 11~kPa (red) have most persistent first modes and least
persistent next higher modes,
\item[3.] cells on 34~kPa (blue) range close to red cells with a clear
tendency (except for the first mode) toward black cells.
\end{longlist}

\section{Discussion} In this research, we have proposed a
nonparametric methodology that can be applied to circular data to
detect differences in shape features of underlying unknown circular
densities. This method builds on a wrapped SiZer (WiZer) because we
found that the wrapped Gaussian kernel is the one and only kernel
(under reasonable assumptions) that preserves circular causality. In a
natural way, with the W/SiZer's methodology comes a notion of shape
feature persistences. We have used this to propose a measure assessing
early differentiation of human mesenchymal stem cells (hMSCs). Our
results warrant larger studies involving larger sample sizes. In such
studies, we expect that inferred mode persistence diagrams (IMPDs) and
combinations of IMPDs with other cell parameters will give a useful
tool to measure early differentiation of hMSCs also over time.

\begin{figure}

\includegraphics{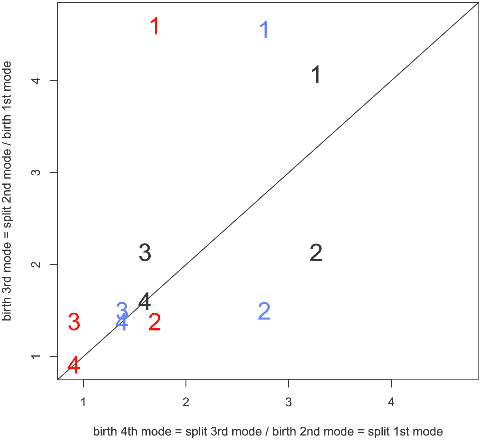}

\caption{Inferred mode-persistence diagram for the first four
modes (cf. Table \protect\ref{persistence:tab}) of the three cells of Figure
\protect\ref{raw-traced:fig} in log-scale. Colors indicate cell
number (cell 1:
black, cell 2: red, cell 3: blue) and coordinates of numbers indicate
(birth, split) for even number of modes (1 and 3) and (split, birth) for
odd number of modes (2 and 4). Odd order modes are more persistent if
they have a higher vertical component, even order are so if their
horizontal component is higher.
(Colors are visible in the online version of
the article.)}\label{persistence:fig}
\end{figure}

\begin{figure}

\includegraphics{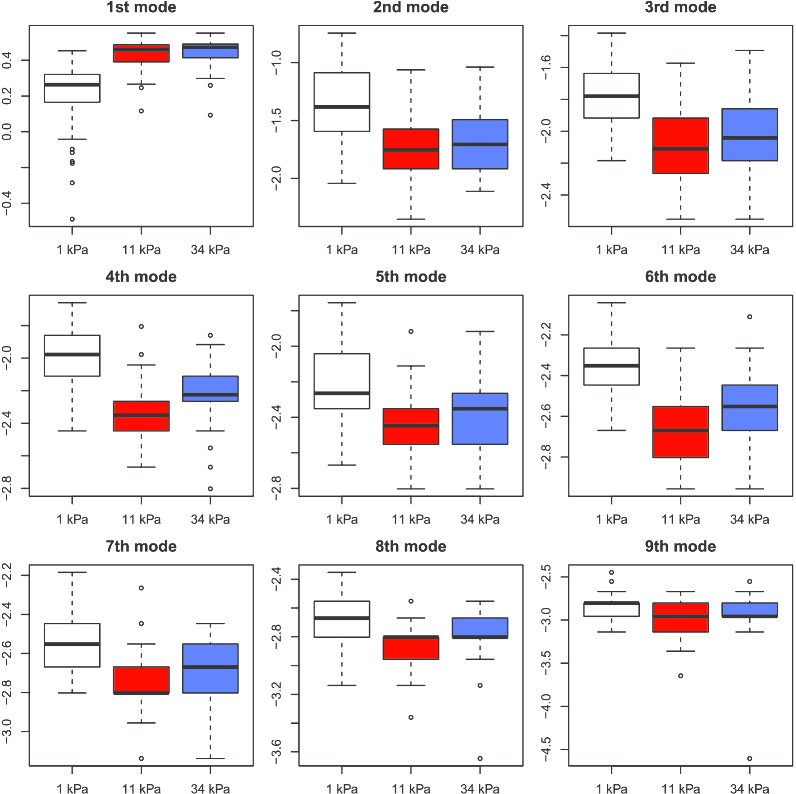}

\caption{Boxplots of inferred mode-persistences in log-scale
radians for the first nine modes for a sample of 179 cells in total
exposed for 24 hours to matrix rigidities of 1~kPa (black, 60 cells),
11~kPa (red, 58 cells) and 34~kPa (blue, 61 cells).
(Colors are visible in the online version of
the article.)}\label
{persistence-sample-boxplots:fig}
\end{figure}

\begin{figure}

\includegraphics{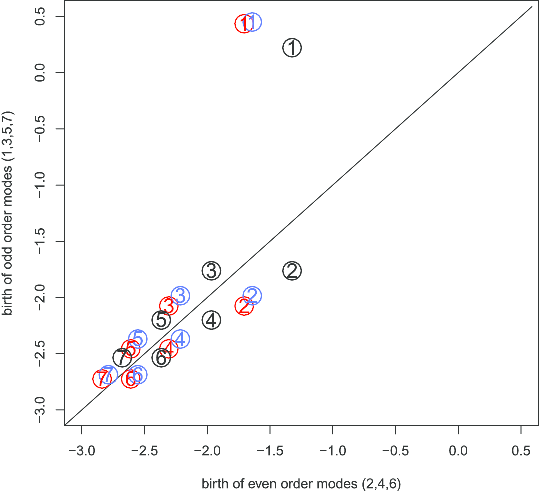}

\caption{Inferred mode-persistence diagram in log-scale radians
for the means of the first seven modes of the data set from Figure
\protect\ref
{persistence-sample-boxplots:fig} with same coloring scheme.
(Colors are visible in the online version of
the article.)}\label
{persistence-sample-means:fig}\vspace*{-8pt}
\end{figure}

The key benefit of the W/SiZer methodology, providing for IMPDs, is
also a key limitation: statistical inference is possible only above a
specific minimal bandwidth $h_0$, which is of course free to choose,
although it may affect the asymptotic approximation, cf. Theorem~\ref
{thmm:Weak_Conv} and Corollary~\ref{weak-conv:cor}. This is due to the
fact that the corresponding asymptotic distributions are not tight as
$h_0\to0$ and a suitable calibration as $h_0 \to0$ and $n\to\infty$
would have to be found that allows to compensate for nontightness (cf.
Chaudhuri and Marron \cite{ChaudhuriMarron2000}, page~420). For the line,
this problem has been
recently solved by Schmidt-Hieber \textit{et al.} \cite{SHMD13}, now allowing for inference on shape
parameters simultaneously over all scales (see also the discussion in
Panaretos \textit{et al.} \cite{Panaretos2013} for a similar reasoning
in the context of flow
estimation). We expect that these arguments can be extended to our setup.

\begin{appendix}\label{app}
\section*{Appendix A: Sign changes}\label{appa}

Let $S^{-}$ denote the usual sign change functional as defined in,
for example, Brown \textit{et al.} \cite{BrownJohnstoneMacGibbon1981},
Definition~2.1, Karlin \cite
{Karlin1968},
that is, for real numbers $x_{1},\ldots,x_{n}$,
the expression
\[
S^{-}(x_{1},\ldots,x_{n})
\]
simply counts the number of changes from positive to negative values
in the sequence $x_{1},\ldots,x_{n}$, or vice versa, ignoring zeros. We
now follow Mairhuber \textit{et al.} \cite{MairhuberSchoenbergWilliamson1959}.

\begin{defn}
The number of \emph{cyclic sign changes} for a vector $\mathbf
{x}=(x_{1},\ldots,x_{n})$
is defined as follows: If $\x=0$ set $S_{c}(\x)=0$. If there
is a nonzero entry, say $x_{j}\neq0$, then set
\[
S_{c}(\x)=S^{-}(x_{j},x_{j+1},
\ldots,x_{n},x_{1},\ldots,x_{j-1},x_{j}).
\]
\end{defn}

\begin{rem}\label{even-sign-changes:rm}
Obviously, this definition does not depend on the choice of the
nonzero entry
$x_{j}$. Moreover, as can be easily seen by induction, the number of
cyclic sign changes is always even.
\end{rem}

We say that a vector $\z=(z_{1},\ldots,z_{n})\in\SS^{n}$ with distinct
entries is in \emph{cyclic order} if there exist real numbers\vspace*{-1pt}
\[
t_{1}<t_{2}<\cdots<t_{n}<t_{1}+2\pi\vspace*{-1pt}
\]
such that $z_{l}=e^{it_{l}}$ for $l=1,\ldots,n$.

\begin{defn}
For a function $f:\SS\ra\RR$ and points $\mathbf{z}=(z_{1},\ldots
,z_{n})\in\SS^{n}$ in cyclic order we set\vspace*{-1pt}
\[
f(\z)=\bigl(f(z_{1}),\ldots,f(z_{n})\bigr).\vspace*{-1pt}
\]
Then we define the number of \emph{cyclic sign changes of $f$} as\vspace*{-1pt}
\[
S_{c}(f)=\mathop{\sup_{n\in\mathbb N, \mathbf{z}\in\SS^n}}
_{\mathrm{in\ cyclic\ order}}S_{c}\bigl(f(
\mathbf{z})\bigr).\vspace*{-1pt}
\]
\end{defn}

\section*{Appendix B: Proofs}\label{appb}
\begin{pf*}{Proof of Theorem \protect\ref{various-axioms:th}}
%

\renewcommand{\theLemm}{B.\arabic{Lemm}}
\begin{Lemm}\label{MR_equiv_VR:lem} \textup{(VD)} $\Leftrightarrow$ \textup{(MR)}.\vadjust{\goodbreak}
\end{Lemm}
\begin{pf}
The implication ``$\Rightarrow$'' is obvious. For
``$\Leftarrow$'', we show the assertion first for a continuous $f\in
L^1(\mathbb S)$. For this, set $F(e^{it}) = \int_0^t f(e^{is}) \,ds$.
Then $ \lim_{t\downarrow0} (F(e^{it})-F(1) )/t = f(1)=\lim_{t\uparrow0} (F(e^{i(2\pi+t)})-F(1) )/t$ and $\mathit{DF}=f$, which
gives as desired,
\[
S_c(L*f) = S_c (L* \mathit{DF} ) 
\leq
S_c (\mathit{DF} )=S_c(f).
\]

Next, assume that $f\in L^1(\mathbb S)$ is arbitrary with $S_c(f) = k
<\infty$ (the assertion is trivial for $k=\infty$). Due to continuity
of $L$, there is a sequence of continuous functions $f_n$ such that
$f_n \to f$ and $L*f_n \to L*f$ in $L^1(\SS)$ as $n\to\infty$. In
fact, because $\SS$ decomposes into $k$ disjoint intervals (some of
which may be points only) in which either $\sign f\geq0$ or $\leq0$,
we may choose the functions $f_n$ such that $S_c(f_n) \to S_c(f)$, that is,
$S_c(f_n) = k$ for all $n\geq n_0$ for some $n_0\in\mathbb N$. Now let
$m\in\mathbb N$ and assume that $S_c(L*f) \geq m$. As $\SS$ is
compact, $L*f$ is uniformly continuous. In consequence, we have that
$S_c(L*f_n)\geq m$ for $n$ sufficiently large. Hence, as all $f_n$ are
continuous, $m\leq k$. In particular, this shows that $S_c(L*f)\leq k$,
concluding the proof.
\end{pf}

To complete the proof of Theorem~\ref{various-axioms:th}, we show the
following.

\begin{Lemm}\label{CV-implies-NE:lem} Suppose that $\{L_h: h>0\}$ is a
differentiable semigroup of circular variation diminishing kernels.
Then it is not enhancing local extrema.
\end{Lemm}

\begin{pf}
Assume that $f$ is some smooth function and $z_{0}$ is an isolated
maximum of the function $z\mapsto L_{h_{0}}*f(z)$ for some $h_{0}>0$.
Let $C=L_{h_{0}}*f(z_{0})$ and $\z=(z_{1},\ldots,z_{k})$ be such
that
\[
S_{c} \bigl(L_{h_{0}}*f(\z)-C \bigr)=S_{c}
(L_{h_{0}}*f-C ).
\]
We find $w_{1},w_{2}\in\SS$ such that
\[
\tilde{\z}=(z_{1},\ldots,w_{1},z_{0},w_{2},
\ldots,z_{k})
\]
is in cyclic order and $L_{h_{0}}*f(w_{j})-C<0$ for $j=1,2$.

Assume now that $\partial_{h}L(z_{0},h_{0})>0$ were true then we
would find $h'>0$ small enough such that
\[
L_{h_{0}+h'}*f(z_{0})-C>0
\]
and also, because $h\mapsto L_{h}*f(z)$ is continuous for fixed $z$,
\[
\sign\bigl(L_{h_{0}+h'}*f(z)-C\bigr)=\sign\bigl(L_{h_{0}}*f(z)-C
\bigr)\qquad \mbox{for }z=w_{1},w_{2},z_{1},
\ldots,z_{n}.
\]
This would imply
\begin{eqnarray*}
S_{c} (L_{h_{0}+h'}*f-C ) \geq S_{c}
\bigl(L_{h_{0}+h'}*f(\tilde {\z})-C \bigr) & =&S_{c}
\bigl(L_{h_{0}}*f(\z)-C \bigr)+2
\\
& =&S_{c} (L_{h_{0}}*f-C )+2,
\end{eqnarray*}
a contradiction, since $L_{h_{0}+h'}*f-C=L_{h'}*
(L_{h_{0}}*f-C )$.
The analog argument works for local minima.
\end{pf}
\noqed\end{pf*}

\begin{pf*}{Proof of Theorem \protect\ref{wrappedGauss-is-causal:thmm}}
With Remark~\ref{gauss-is-regular:rm}, we are left to show that the
wrapped Gaussian is variation diminishing. We proceed as follows.
First, we note that the Gaussian kernels give an exponential family on
the line and that exponential families give rise to variation
diminishing kernels on the line (cf. Brown \textit{et al.}
\cite{BrownJohnstoneMacGibbon1981}, Example~3.1). Here, we use \emph{variation
diminishing} on $\RR$ as introduced by Karlin \cite{Karlin1968} and
Brown \textit{et al.} \cite
{BrownJohnstoneMacGibbon1981}. Secondly, note that kernels made from
exponential families on the line are \emph{mean continuous} in the
following sense.

A function $K:\RR\ra[0,\infty)$ with $\int_\RR K(t)  \,dt=1$ is a
\emph
{mean continuous kernel} if
\[
2K(t) = \lim_{h\downarrow0} \bigl(K(t+h) + K(t-h) \bigr)\qquad\mbox{for all
}t\in\mathbb R.
\]

Then the assertion of Theorem~\ref{wrappedGauss-is-causal:thmm} follows
from the following.

\renewcommand{\theThhh}{B.3}
\begin{Thhh}\label{exp-fam-wrapping:thmm} Let $K$ be a mean continuous
kernel that is variation diminishing on the line. Then its \emph
{wrapped kernel}
\[
\widetilde{K}\bigl(e^{it}\bigr) = \sum_{k\in\mathbb Z}K(t+2
\pi)
\]
is variation diminishing on the circle.
\end{Thhh}

\begin{pf}
This is a consequence of the following two theorems of Mairhuber
\textit{et al.} \cite
{MairhuberSchoenbergWilliamson1959}. To this end recall the
Laguerre--P\'olya class of entire functions $\Psi:\CC\ra\CC$ of form
\[
\Psi(s)=e^{-\alpha s^{2}+\beta_{0}s}\prod_{k=1}^{\infty}(1+
\beta _{k}s)e^{-\beta_{k}s} ,
\]
where $\alpha\geq0$ and $\beta_{k}$ are real such that
\[
0<\alpha+\sum_{k}\beta^2_{k}<
\infty.
\]
The first theorem connects functions in
the Laguerre--P\'olya class to mean continuous variation diminishing
kernels on the line.

\renewcommand{\theThhhh}{B.4}
\begin{Thhhh}[(Mairhuber \textit{et al.} \cite{MairhuberSchoenbergWilliamson1959}, Theorem~B)]\label
{thmm:causality}
If $\Psi$ belongs to the Laguerre--P\'olya
class then there exists an $\ve>0$ such that $\nicefrac{1}{\Psi}$
is holomorphic in the open strip
%
\renewcommand{\theequation}{\arabic{equation}}
\setcounter{equation}{11}
\begin{equation}
S=\bigl\{z\in\CC:-\ve<\Re(z)<\ve\bigr\}.\label{eq:OpenStrip}
\end{equation}
Further, there exists a mean continuous variation diminishing kernel
$\Lambda:\RR\ra\RR$
such that
%
\begin{equation}
\frac{1}{\Psi(z)}=\int_{-\infty}^{\infty}e^{-zt}
\Lambda (t)\,dt\label{eq:1/Psi}
\end{equation}
holds for all $z\in S$.
Conversely, if $\Lambda:\RR\ra\RR$ is a mean continuous variation
diminishing kernel
then there exists a function $\Psi$ in the Laguerre--P\'olya class such
that equation (\ref{eq:1/Psi}) holds in an open strip of the form
(\ref{eq:OpenStrip}).
\end{Thhhh}

The second theorem completes the proof of Theorem~\ref
{exp-fam-wrapping:thmm} as it links Laguerre--P\'olya functions to
circular variation diminishing kernels via wrapping.

\renewcommand{\theThhhhhh}{B.5}
\begin{Thhhhhh}[(Mairhuber \textit{et al.} \cite{MairhuberSchoenbergWilliamson1959}, Theorem~3)] Every
element $\Psi
$ of the Laguerre--P\'olya
class gives rise to a mean continuous cyclic variation diminishing
kernel via
\[
\Omega\bigl(e^{it}\bigr)=\sum_{k\in\ZZ}
\frac{1}{\Psi(ik)}e^{ikt}.
\]
If $\Lambda(t)$ is a function for which (\ref{eq:1/Psi}) holds then
we have the identity
\[
\hspace*{100pt}\Omega\bigl(e^{it}\bigr)=2\pi\sum_{k\in\ZZ}
\Lambda(t+2\pi k).\hspace*{150pt}\qed
\]
\end{Thhhhhh}
\noqed\end{pf}
\noqed\end{pf*}

\begin{pf*}{Proof of Theorem \protect\ref{causal-is-wrappedGauss:thmm}}
In this proof, we follow the argument of Lindeberg \cite
{Lindeberg2011} starting
with the following lemma which is a special case of Engel and Nagel \cite{EngelNagel2000}, Lemma~1.3.
and Theorem~1.4. To this end, recall the property of
strong continuity (SC) introduced in Remark~\ref{strong-continuity:rm}.

\renewcommand{\theLemmm}{B.6}
\begin{Lemmm}\label{almost-heat-eqn:lem}
For a strongly continuous semigroup of kernels $L_{h}:\SS\ra\RR$,
$h> 0$ there is a 
(not necessarily bounded) closed
linear operator $\mathcal{A}:L^{2}(\SS)\ra L^{2}(\SS)$ defined on a
dense domain $\mathcal{D}(\mathcal{A})\subset L^{2}(\SS)$, such that
for $f\in\mathcal{D}(\mathcal{A})$ we have
\[
\mathcal{A}f=\lim_{h\downarrow0}\frac{L_{h}*f-f}{h} ,
\]
as well as for all $h> 0$ that
%
\renewcommand{\theequation}{\arabic{equation}}
\setcounter{equation}{13}
\begin{equation}
\label{pre-heat:eqn} \partial_h (L_h*f) = \mathcal
{A}(L_h*f).
\end{equation}
\end{Lemmm}

The operator $\mathcal{A}$ is called the \emph{infinitesimal generator}
of the semigroup $\{ L_{h}:h> 0\}$.
By an argument analogous to that of Remark~\ref{rem:diffSemiGroup}, we
see $C^\infty(\SS)\subset\mathcal{D}(\mathcal{A})$.

For $w\in\SS$ let $\Delta_{w}:C^{\infty}(\SS)\ra C^{\infty}(\SS)$
be the shift operator defined by $\Delta_{w}f(z)=f(w^{-1}z)$. Then
it is easy to see that $\mathcal{A}$ commutes with $\Delta_{w}$
since
\[
\Delta_{w} \biggl(\frac{L_{h}*f-f}{h} \biggr)=\frac{L_{h}*(\Delta
_{w}f)-\Delta_{w}f}{h}.
\]
Moreover, the identity $R(L_{h}*f-f)=L_{h}*(Rf)-Rf$ (by symmetry of
$L_h$) implies that
$\mathcal{A}$ also commutes with the reflection operator $R$, defined
via $Rf(z)=f(z^{-1})$.

In the next step, we show that $\mathcal{A}$ is a differential
operator. To this end, we exploit Peetre \cite{Peetre1959}, Th\'eor\`eme,
stating that $\mathcal{A}$ is a differential operator if and only if it
has the following property (of a \emph{local} operator)
\[
\begin{tabular}{p{300pt}@{}}
$\mathcal{A}f(z_0)=0$ for any $z_0\in\SS$ and smooth function
$f:\SS
\to\RR$ vanishing in a neighborhood $U$ of $z_0$.
\end{tabular}
\]
Indeed, assuming that the smooth function $f:\SS\to\RR$ vanishes in a
neighborhood $U$ of $z_0$, let $g:\SS\to\RR$ be smooth, having a
nondegenerate maximum at $z_0$ with support in $U$. In consequence,
for every $\beta\in\RR$ the function $r= \beta f + g$ has a
nondegenerate maximum at $z_0$, and hence $\partial_h|_{h=0}
(L_h*r)(z_0) \leq0$ by Lemma~\ref{CV-implies-NE:lem} which implies by
Lemma~\ref{almost-heat-eqn:lem} that $\mathcal{A}(r)(z_0)\leq0$. By
linearity,
\[
\beta\mathcal{A}(f) (z_0) + \mathcal{A}(g) (z_0)
\leq0,
\]
where the left-hand side can be given any sign (as $\beta\in\RR$)
unless $\mathcal{A}(f)(z_0)=0$. In consequence, we have shown
%
\renewcommand{\theequation}{\arabic{equation}}
\setcounter{equation}{14}
\begin{equation}
\mathcal{A}=\sum_{0\leq k\in\ZZ}a_{k}D^{k}\label{eq:Af}
\end{equation}
with suitable $a_k \in\RR$, ($k=0,1,\ldots$) at least in each chart.
Using $\mathcal{A}\Delta_{w}=\Delta_{w}\mathcal{A}$
we can conclude that (\ref{eq:Af}) holds globally.

In the penultimate step, we show that $a_{k}=0$ in (\ref{eq:Af}) for $k>2$.
Let $f$ be a smooth function on $\SS$ such that
\[
f\bigl(\exp(it)\bigr)=t^{2}+\beta t^{n}
\]
for $t$ in some neighborhood of zero and for some integer $n>2$.
Since $f$ has an isolated local minimum at $1$ for all $\beta$
we have by Lemmata \ref{CV-implies-NE:lem} and \ref
{almost-heat-eqn:lem} that $\partial_h|_{h=0} (L_h*f)(1)\geq0$
(arguing as above), and hence $\mathcal{A}f(1)\geq0$. On the other
hand, we have $\mathcal{A}f(1)=a_{n}\beta$ and hence $a_{n}\neq0$
would give a contradiction.

If we set $n=0$ in the above argument, we conclude that $a_{0}=0$.

Finally, the identity $R\mathcal{A}=\mathcal{A}R$
yields $a_{1}=0$ and, therefore, $\mathcal{A}=a_{2}D^{2}$. In
consequence, for $f\in C^\infty(\SS)$, (\ref{pre-heat:eqn}) is a
multiple of the heat equation (\ref{heat:eq}) with initial condition
$f$, which proves
the theorem.
\end{pf*}

\begin{pf*}{Proof of Theorem \protect\ref{thmm:circ-causality}}
In order to see monotonicity, let $g_{h}(z)$ be either $D^{m}f_{h}(z)$
or $D^{m}f_{h}^{(n)}(z)$
for $h>0$ corresponding to case (i) and for case (ii) set
$g_{0}(z)=D^{m}f(z)$. Now
let $h_{2}>h_{1}>0$ in case (i) or $h_{2}>h_{1}\geq0$ in case (ii). Then
$g_{h_{2}}=K_{h_{2}-h_{1}}*g_{h_{1}}$
and hence
%
\renewcommand{\theequation}{\arabic{equation}}
\setcounter{equation}{15}
\begin{equation}
S_{c}(g_{h_{2}})\leq S_{c}(g_{h_{1}})
\label{monotonicity:ineq}
\end{equation}
since $K_{h}$ is cyclic variation diminishing for all $h>0$ by Theorem~\ref{causal-is-wrappedGauss:thmm}.

To see right continuity, suppose that $S_{c}(g_{h})=2k$ for integer
$k\geq0$ and some $h>0$ in case (i)
or $h\geq0$ in case (ii).
If $k>0$, there exists $\z=(z_{1},\ldots,z_{2k})\in\SS^{2k}$
such that $\sgn g_{h}(z_{j})=(-1)^{j}$. Since $h\mapsto g_{h}(z)$
is continuous for all $z\in\SS$ there is an $\varepsilon>0$ such that
$\sgn g_{h'}(z_{j})=(-1)^{j}$ for all $h'\in[h,h+\varepsilon)$.
Therefore,
\[
S_{c}(g_{h'})\geq S_{c} \bigl(g_{h'}(
\z) \bigr)=2k.
\]
Together with the monotonicity (\ref{monotonicity:ineq}), this proves
$S_{c}(g_{h'})=2k$, and thus right continuity.
If $k=0$, we have at once $S_{c}(g_{h'})\geq0 = S_{c}
(g_{h}
)$ for all $h'\geq h$ again yielding right continuity.
\end{pf*}

\begin{pf*}{Proof of Theorem \protect\ref{thmm:Weak_Conv}}
We first show convergence of the finite dimensional distributions (I),
then weak convergence (II), and finally a.s. continuity of sample paths.

I: Let $z_{1},\ldots,z_{k}\in\SS$ and $t_{1},\ldots,t_{k}\in\RR$
and set
\[
Y_{l}=\sum_{i=1}^{k}t_{i}
\bigl(D^{m}K_{h}\bigl(z_{i}X_{l}^{-1}
\bigr)-D^{m}f_{h}(z_{i}) \bigr),\qquad l=1,\ldots,n.
\]
Then $EY_{l}=0$ and
\[
\var Y_{l}  =\frac{1}{n}\sum_{i,j=1}^{k}t_{i}t_{j}
\cv(z_{i},z_{j};h).
\]
In consequence, we have that
\[
Z_{n}  :=n^{{1}/{2}}\sum_{i=1}^{k}t_{i}
\bigl(D^{m}f_{h}^{(n)}(z_{i})-D^{m}f_{h}(z_{i})
\bigr) =n^{-
{1}/{2}}\sum_{l=1}^{n}Y_{l}
\]
is asymptotically normal with zero mean and covariance $\sum_{i,j=1}^{k}t_{i}t_{j}\cv(z_{i},z_{j};h)$.
Since $t_{1},\ldots,t_{k}\in\RR$ are arbitrary, by the Cram\'er--Wold
device we conclude that
\[
U_{n}(z_{1},\ldots,z_k):= \bigl(
n^{{1}/{2}} \bigl(D^{m}f_{h}^{(n)}(z_{i})-D^{m}f_{h}(z_{i})
\bigr) \bigr) _{i=1}^{k}
\]
is asymptotically normal with zero mean and covariance matrix $
(\cv
(z_{i},z_{j};h) )_{i,j=1}^k$.

II: To establish weak convergence, it now suffices to show that
\[
n^{{1}/{2}} \bigl(D^{m}f_{h}^{(n)}(z)-D^{m}f_{h}(z)
\bigr)_{z\in
\SS}
\]
is tight. To this end, we give a bound for the second moments
of the increments. Observe that for $z_{1},z_{2}\in\SS$,
%
\renewcommand{\theequation}{\arabic{equation}}
\setcounter{equation}{16}
\begin{eqnarray}
\label{eq:Estimate_U_n} E\bigl(U_{n}(z_{1})-U_{n}(z_{2})
\bigr)^{2} & =&\var \bigl\{ D^{m}K_{h}
\bigl(z_{1}X^{-1}\bigr)-D^{m}\bigl(K_{h}
\bigl(z_{2}X^{-1}\bigr)\bigr) \bigr\}
\nonumber
\\
& \leq& E \bigl\{ D^{m}K_{h}\bigl(z_{1}X^{-1}
\bigr)-D^{m}\bigl(K_{h}\bigl(z_{2}X^{-1}
\bigr)\bigr) \bigr\} ^{2}
\\
& \leq&\bigl\Vert D^{m+1}K_{h}\bigr\Vert ^2_{\infty}\,d(z_{1},z_{2})^{2}\nonumber
\end{eqnarray}
using the mean value theorem for the last inequality where $d$ denotes
the intrinsic (geodesic) distance in $\SS$. Now we argue with Theorem~8.6 (page~138) and Theorem~8.8 (page~139) of Ethier and Kurtz \cite
{EthierKurtz2009}. First
note that (\ref{eq:Estimate_U_n}) yields condition (8.36) of Theorem~8.8. which asserts condition (8.29) of Theorem~8.6 as well as condition
(8.39). Taking the conditional expectation of the latter gives
condition (8.28) of Theorem~8.6. Directly from (\ref{eq:Estimate_U_n}),
we infer the third condition (8.30) of (b) of Theorem~8.6. From the fact
that in particular for every fixed $z\in\SS$, $U_n(z)$ is
asymptotically normal, we obtain condition (a) of Theorem~7.2 (Ethier
and Kurtz \cite{EthierKurtz2009}, page~128). Thus, we conclude with Theorem~8.6. that the
sequence $U_n(z)$, $n\in\mathbb N$ of processes $z\in\SS$ is
relatively compact, yielding the desired convergence result. 

III: Since $G_h$ is the weak limit of a sequence of a.s. continuous
processes, there is a version of $G_h$ that also has continuous sample
paths with probability one.
\begin{pf*}{Proof of Theorem \protect\ref{thmm:TestThm}}
Let $\mathcal{S}_{0}\subset\{(h,z):h\geq h_{0},z\in\SS\}$ for which
$H_{0}^{(h,z)}$ holds. Then by Corollary~\ref{weak-conv:cor},
the probability that one or more of the
null hypotheses are falsely rejected is
\begin{eqnarray*}
&& P \bigl(\exists(h,z)\in\mathcal{S}_{0}\mbox{ s.t. }\bigl\llvert
D^{m}f_{h}^{(n)}(z)\bigr\rrvert \geq
n^{-{1}/{2}}q_{(1-\alpha
)} \bigr)
\\
&&\quad \leq P \Bigl(\sup_{h\geq h_{0},z\in\SS}n^{{1}/{2}}\bigl\llvert
D^{m}f_{h}^{(n)}(z)-D^{m}f_{h}(z)
\bigr\rrvert \geq q_{(1-\alpha)} \Bigr)
\\
&&\quad \ra P \Bigl(\sup_{z\in\SS}|G_{h_{0}}|\geq
q_{(1-\alpha)} \Bigr)=\alpha \qquad\mbox{as }n\to\infty.
\end{eqnarray*}

For the last assertion, note that if $D^{m}f_{h}(z)$ has $2k\geq2$
sign changes then we can find a vector $(z_{1},\ldots,z_{2k})$ in
cyclic order such that $\sgn D^{m}f_{h}(z_j)=(-1)^{j}$ ($j=1,\ldots
,2k$). Because of
$q_{(1-\alpha)}n^{-{1}/{2}}\ra0$ and
\begin{eqnarray*}
\max_{1\leq j\leq2k}\bigl\llvert D^{m}f_{h}^{(n)}(z_{j})-D^{m}f_{h}(z_{j})
\bigr\rrvert  \leq\sup_{z\in
\SS
}\bigl\llvert D^{m}f_{h}^{(n)}(z)-D^{m}f_{h}(z)
\bigr\rrvert =\mathcal {O}_{P}\bigl(n^{-{1}/{2}}\bigr)
\end{eqnarray*}
the above test procedure will correctly conclude the sign of $D^{m}f_{h}(z_{j})$
for all $j$ asymptotically with probability one.
\end{pf*}
\noqed\end{pf*}

\section*{Appendix C: Numerical considerations}\label{appc}

Let $K:\mathbb{R}\ra\mathbb{R}_{\geq0}$ be a kernel on the real axis
and $\tilde{K}:[-\pi,\pi]\ra\mathbb{R}_{\geq0}$ the corresponding
wrapped kernel defined by $\tilde{K}(x)=\sum_{j\in\ZZ}K(x+2\pi j)$
and for $C\in\ZZ$ the cut off kernel $\tilde{K}^{(C)}(x)=\sum_{|j|\leq
C}K(x+2\pi j)$.

For points $x,x_{1},\ldots, x_{n}\in[-\pi,\pi]$, we define
\[
f(x) =\frac{1}{n}\sum_{i=1}^{n}
\tilde{K}(x-x_{i}),\qquad f^{(C)}(x) =\frac{1}{n}\sum
_{i=1}^{n}\tilde{K}^{(C)}(x-x_{i}).
\]
We estimate
\begin{eqnarray*}
&&\bigl|f(x)-f^{(C)}(x)\bigr|
\\
&&\quad =\Biggl|\frac{1}{ n}\sum_{i=1}^{n}\sum
_{|j|>C}K(x-x_{i}+2\pi j)\Biggr| \leq
\frac{1}{n}\sum_{i=1}^{n}\sum
_{|j|>C}\bigl|K(x-x_{i}+2\pi j)\bigr|
\\
&&\quad \leq\sum_{j> C} \Bigl( \sup_{x\in[2\pi(j-1) ,2\pi(j+1)]}
\bigl|K(x)\bigr|+\sup_{x\in[-2\pi(j+1),-2\pi(j-1)]}\bigl|K(x)\bigr| \Bigr) =: \delta_{K}(C).
\end{eqnarray*}
Now we want to estimate $\delta_{K}(C)$ for the derivative of the
Gaussian kernel
\[
\phi_{h}(x)=\frac{1}{\sqrt{2\pi h}}e^{-{x^{2}}/{(2h)}}.
\]
Note that $\partial_{x}\phi_{h}$ has two extremal points at $\pm
\sqrt{h}$
and is monotonically increasing on $(-\infty,-\sqrt{h})$ and monotonically
decreasing on $(\sqrt{h},\infty)$. Hence, if $2\pi(C-1)\geq\sqrt{h}$
we have
\[
\sup_{x\in[-2\pi(j+1),-2\pi(j-1)]\cup[2\pi(j-1),2\pi
(j+1)]}\bigl|\partial _{x}\phi_h(x)\bigr|
\leq- \frac{1}{2\pi}\int_{ 2\pi(j-2)}^{2\pi
(j-1)}
\partial_{x}\phi_h(y) \,dy
\]
for all $j\geq C+1$ and, therefore,
\begin{eqnarray*}
\delta_{\partial_{x}\phi_h}(C)  \leq- \frac{1}{\pi}\sum
_{j\geq C+1} \int_{2\pi(j-2)}^{2\pi
(j-1)}
\partial_{x}\phi_{h}(y) \,dy =\frac{1}{\pi}
\phi_{h}\bigl(2\pi(C-2)\bigr).
\end{eqnarray*}

\renewcommand{\thethmmmmm}{C.1}
\begin{thmmmmm}
For any empirical measure $\nu=\frac{1}{n}\sum_{i=1}^{n}\delta_{x_{i}}$
with points $x_{1},\ldots,x_{n}\in[-\pi,\pi]$, we have
\[
\biggl\Vert \frac{\partial(\tilde{\phi}_{h}*\nu)}{\partial x}-\frac
{\partial
(\tilde{\phi}_{h}^{(C)}*\nu)}{\partial x}\biggr\Vert _{\infty}\leq
\frac
{1}{\pi
}\phi_{h}\bigl(2\pi(C-1)\bigr)
\]
provided that $2\pi(C-1)\geq\sqrt{h}$.
\end{thmmmmm}

In consequence, for the bandwidths $0.01 \leq h \leq3.5$ considered in
our applications in Sections~\ref{persistence:scn} and \ref
{application:scn}, an error less than
\[
\frac{e^{-10 }}{\pi\sqrt{2\pi h}}\leq10^{-4}
\]
can be achieved for $2\pi(C-1) \geq\sqrt{20 h}$, for example, $C=3$.
\end{appendix}

\section*{Acknowledgments} All authors gratefully acknowledge
support by DFG SFB 755 ``Nanonscale Photonic Imaging'' Projects B8, A4
and A6. Furthermore, M. Sommerfeld is grateful for support by the
``Studienstiftung des Deutschen Volkes''. A. Munk also gratefully
acknowledges support by DFG FOR 916. F. Rehfeldt gratefully acknowledges
funding through a Feodor--Lynen fellowship from the
Alexander-von-Humboldt foundation. S. Huckemann and A. Munk are further
indebted
to the Volkswagen Stiftung. Finally, M.~Sommerfeld and S. Huckemann are thankful for
support from the SAMSI 2013--2014 program on Low-Dimensional Structure
in High-Dimensional Systems.

%

%




\printhistory

\begin{thebibliography}{53}

\bibitem{AhmedWalter2012}
%
\begin{barticle}[mr]
\bauthor{\bsnm{Ahmed},~\bfnm{Murat~O.}\binits{M.O.}} \AND
\bauthor{\bsnm{Walther},~\bfnm{Guenther}\binits{G.}}
(\byear{2012}).
\btitle{Investigating the multimodality of multivariate data with
principal curves}.
\bjournal{Comput. Statist. Data Anal.}
\bvolume{56}
\bpages{4462--4469}.
\bid{doi={10.1016/j.csda.2012.02.020}, issn={0167-9473}, mr={2957886}}
\end{barticle}
%

\bptok{imsref}%
\endbibitem

\bibitem{AGLM93}
%
\begin{barticle}[mr]
\bauthor{\bsnm{Alvarez},~\bfnm{Luis}\binits{L.}},
\bauthor{\bsnm{Guichard},~\bfnm{Fr{\'e}d{\'e}ric}\binits{F.}},
\bauthor{\bsnm{Lions},~\bfnm{Pierre-Louis}\binits{P.-L.}} \AND
\bauthor{\bsnm{Morel},~\bfnm{Jean-Michel}\binits{J.-M.}}
(\byear{1993}).
\btitle{Axioms and fundamental equations of image processing}.
\bjournal{Arch. Ration. Mech. Anal.}
\bvolume{123}
\bpages{199--257}.
\bid{doi={10.1007/BF00375127}, issn={0003-9527}, mr={1225209}}
\end{barticle}
%

\bptok{imsref}%
\endbibitem

\bibitem{BWBD86}
%
\begin{barticle}[auto:parserefs-M02]
\bauthor{\bsnm{Babaud},~\bfnm{J.}\binits{J.}},
\bauthor{\bsnm{Witkin},~\bfnm{A.~P.}\binits{A.P.}},
\bauthor{\bsnm{Baudin},~\bfnm{M.}\binits{M.}} \AND
\bauthor{\bsnm{Duda},~\bfnm{R.~O.}\binits{R.O.}}
(\byear{1986}).
\btitle{Uniqueness of the {G}aussian kernel for scale space filtering}.
\bjournal{IEEE Transactions on Pattern Analysis and Machine Intelligence}
\bvolume{8}
\bpages{26--33}.
\end{barticle}
%

\bptok{imsref}%
\endbibitem

\bibitem{BalakrishnanFasyLecciRinaldiAartiWasserman2013}
%
\begin{bmisc}[auto:parserefs-M02]
\bauthor{\bsnm{Balakrishnan},~\bfnm{S.}\binits{S.}},
\bauthor{\bsnm{Fasy},~\bfnm{B.}\binits{B.}},
\bauthor{\bsnm{Lecci},~\bfnm{F.}\binits{F.}},
\bauthor{\bsnm{Rinaldo},~\bfnm{A.}\binits{A.}},
\bauthor{\bsnm{Singh},~\bfnm{A.}\binits{A.}} \AND
\bauthor{\bsnm{Wasserman},~\bfnm{L.}\binits{L.}}
(\byear{2013}).
\bhowpublished{%
Statistical inference for persistent homology.
Preprint. Available at
\arxivurl{arXiv:1303.7117}.}
\end{bmisc}
%

\bptok{imsref}%
\endbibitem

\bibitem{BDMS04}
%
\begin{binproceedings}[auto:parserefs-M02]
\bauthor{\bsnm{Briggs},~\bfnm{A.~J.}\binits{A.J.}},
\bauthor{\bsnm{Detweiler},~\bfnm{C.}\binits{C.}},
\bauthor{\bsnm{Mullen},~\bfnm{P.~C.}\binits{P.C.}} \AND
\bauthor{\bsnm{Scharstein},~\bfnm{D.}\binits{D.}}
(\byear{2004}).
\btitle{May. scale-space features in {1D} omnidirectional images}.
In \bbooktitle{Proc. Fifth Workshop on Omnidirectional Vision, Camera
Networks and Nonclassical Cameras}
(\beditor{\bfnm{P.}\binits{P.}~\bsnm{Sturm}},
\beditor{\bfnm{T.}\binits{T.}~\bsnm{Svoboda}} \AND
\beditor{\bfnm{S.}\binits{S.}~\bsnm{Teller}}, eds.)
\bpages{115--126}.
\blocation{Prague, Czech Republic}.
\end{binproceedings}
%

\bptok{imsref}%
\endbibitem

\bibitem{BrownJohnstoneMacGibbon1981}
%
\begin{barticle}[mr]
\bauthor{\bsnm{Brown},~\bfnm{Lawrence~D.}\binits{L.D.}},
\bauthor{\bsnm{Johnstone},~\bfnm{Iain~M.}\binits{I.M.}} \AND
\bauthor{\bsnm{MacGibbon},~\bfnm{K.~Brenda}\binits{K.B.}}
(\byear{1981}).
\btitle{Variation diminishing transformations: A~direct approach to
total positivity and its statistical applications}.
\bjournal{J. Amer. Statist. Assoc.}
\bvolume{76}
\bpages{824--832}.
\bid{issn={0162-1459}, mr={0650893}}
\end{barticle}
%

\bptok{imsref}%
\endbibitem

\bibitem{BubenikKim2007}
%
\begin{barticle}[mr]
\bauthor{\bsnm{Bubenik},~\bfnm{Peter}\binits{P.}} \AND
\bauthor{\bsnm{Kim},~\bfnm{Peter~T.}\binits{P.T.}}
(\byear{2007}).
\btitle{A statistical approach to persistent homology}.
\bjournal{Homology, Homotopy Appl.}
\bvolume{9}
\bpages{337--362}.
\bid{issn={1532-0073}, mr={2366953}}
\end{barticle}
%

\bptok{imsref}%
\endbibitem

\bibitem{Carlsson2009}
%
\begin{barticle}[mr]
\bauthor{\bsnm{Carlsson},~\bfnm{Gunnar}\binits{G.}}
(\byear{2009}).
\btitle{Topology and data}.
\bjournal{Bull. Amer. Math. Soc. (N.S.)}
\bvolume{46}
\bpages{255--308}.
\bid{doi={10.1090/S0273-0979-09-01249-X}, issn={0273-0979}, mr={2476414}}
\end{barticle}
%

\bptok{imsref}%
\endbibitem

\bibitem{ChaudhuriMarron1999}
%
\begin{barticle}[mr]
\bauthor{\bsnm{Chaudhuri},~\bfnm{Probal}\binits{P.}} \AND
\bauthor{\bsnm{Marron},~\bfnm{J.~S.}\binits{J.S.}}
(\byear{1999}).
\btitle{Si{Z}er for exploration of structures in curves}.
\bjournal{J. Amer. Statist. Assoc.}
\bvolume{94}
\bpages{807--823}.
\bid{doi={10.2307/2669996}, issn={0162-1459}, mr={1723347}}
\end{barticle}
%

\bptok{imsref}%
\endbibitem

\bibitem{ChaudhuriMarron2000}
%
\begin{barticle}[mr]
\bauthor{\bsnm{Chaudhuri},~\bfnm{Probal}\binits{P.}} \AND
\bauthor{\bsnm{Marron},~\bfnm{J.~S.}\binits{J.S.}}
(\byear{2000}).
\btitle{Scale space view of curve estimation}.
\bjournal{Ann. Statist.}
\bvolume{28}
\bpages{408--428}.
\bid{doi={10.1214/aos/1016218224}, issn={0090-5364}, mr={1790003}}
\end{barticle}
%

\bptok{imsref}%
\endbibitem

\bibitem{ChengHall1999}
%
\begin{barticle}[mr]
\bauthor{\bsnm{Cheng},~\bfnm{Ming-Yen}\binits{M.-Y.}} \AND
\bauthor{\bsnm{Hall},~\bfnm{Peter}\binits{P.}}
(\byear{1999}).
\btitle{Mode testing in difficult cases}.
\bjournal{Ann. Statist.}
\bvolume{27}
\bpages{1294--1315}.
\bid{doi={10.1214/aos/1017938927}, issn={0090-5364}, mr={1740110}}
\end{barticle}
%

\bptok{imsref}%
\endbibitem

\bibitem{ChungBubenikKim2009}
%
\begin{bincollection}[auto:parserefs-M02]
\bauthor{\bsnm{Chung},~\bfnm{M.~K.}\binits{M.K.}},
\bauthor{\bsnm{Bubenik},~\bfnm{P.}\binits{P.}} \AND
\bauthor{\bsnm{Kim},~\bfnm{P.~T.}\binits{P.T.}}
(\byear{2009}).
\btitle{Persistence diagrams of cortical surface data}.
In \bbooktitle{Information Processing in Medical Imaging}
\bpages{386--397}.
\blocation{Berlin-Heidelberg}: \bpublisher{Springer}.
\end{bincollection}
%

\bptok{imsref}%
\endbibitem

\bibitem{CohenEdelsbrunnerHarer2007}
%
\begin{barticle}[mr]
\bauthor{\bsnm{Cohen-Steiner},~\bfnm{David}\binits{D.}},
\bauthor{\bsnm{Edelsbrunner},~\bfnm{Herbert}\binits{H.}} \AND
\bauthor{\bsnm{Harer},~\bfnm{John}\binits{J.}}
(\byear{2007}).
\btitle{Stability of persistence diagrams}.
\bjournal{Discrete Comput. Geom.}
\bvolume{37}
\bpages{103--120}.
\bid{doi={10.1007/s00454-006-1276-5}, issn={0179-5376}, mr={2279866}}
\end{barticle}
%

\bptok{imsref}%
\endbibitem

\bibitem{DaviesKovac2004}
%
\begin{barticle}[mr]
\bauthor{\bsnm{Davies},~\bfnm{P.~Laurie}\binits{P.L.}} \AND
\bauthor{\bsnm{Kovac},~\bfnm{Arne}\binits{A.}}
(\byear{2004}).
\btitle{Densities, spectral densities and modality}.
\bjournal{Ann. Statist.}
\bvolume{32}
\bpages{1093--1136}.
\bid{doi={10.1214/009053604000000364}, issn={0090-5364}, mr={2065199}}
\bptnote{check volume}%
\end{barticle}
%

\bptok{imsref}%
\endbibitem

\bibitem{DumbgenSpokoiny2001}
%
\begin{barticle}[mr]
\bauthor{\bsnm{D{\"u}mbgen},~\bfnm{Lutz}\binits{L.}} \AND
\bauthor{\bsnm{Spokoiny},~\bfnm{Vladimir~G.}\binits{V.G.}}
(\byear{2001}).
\btitle{Multiscale testing of qualitative hypotheses}.
\bjournal{Ann. Statist.}
\bvolume{29}
\bpages{124--152}.
\bid{doi={10.1214/aos/996986504}, issn={0090-5364}, mr={1833961}}
\end{barticle}
%

\bptok{imsref}%
\endbibitem

\bibitem{DumbgenWalther2008}
%
\begin{barticle}[mr]
\bauthor{\bsnm{D{\"u}mbgen},~\bfnm{Lutz}\binits{L.}} \AND
\bauthor{\bsnm{Walther},~\bfnm{G{\"u}nther}\binits{G.}}
(\byear{2008}).
\btitle{Multiscale inference about a density}.
\bjournal{Ann. Statist.}
\bvolume{36}
\bpages{1758--1785}.
\bid{doi={10.1214/07-AOS521}, issn={0090-5364}, mr={2435455}}
\bptnote{check volume}%
\end{barticle}
%

\bptok{imsref}%
\endbibitem

\bibitem{EdelsbrunnerLetscherZomorodian2002}
%
\begin{barticle}[mr]
\bauthor{\bsnm{Edelsbrunner},~\bfnm{Herbert}\binits{H.}},
\bauthor{\bsnm{Letscher},~\bfnm{David}\binits{D.}} \AND
\bauthor{\bsnm{Zomorodian},~\bfnm{Afra}\binits{A.}}
(\byear{2002}).
\btitle{Topological persistence and simplification}.
\bjournal{Discrete Comput. Geom.}
\bvolume{28}
\bpages{511--533}.
\bid{doi={10.1007/s00454-002-2885-2}, issn={0179-5376}, mr={1949898}}
\end{barticle}
%

\bptok{imsref}%
\endbibitem

\bibitem{EltznerWollnikGottschlichHuckemannRehfeldt2014}
%
\begin{barticle}[auto:parserefs-M02]
\bauthor{\bsnm{Eltzner},~\bfnm{B.}\binits{B.}},
\bauthor{\bsnm{Gottschlich},~\bfnm{C.}\binits{C.}},
\bauthor{\bsnm{Wollnik},~\bfnm{C.}\binits{C.}},
\bauthor{\bsnm{Huckemann},~\bfnm{S.}\binits{S.}} \AND
\bauthor{\bsnm{Rehfeldt},~\bfnm{F.}\binits{F.}}
(\byear{2015}).
\btitle{The filament sensor for near real-time detection of cytoskeletal fiber
structures}.
\bjournal{PLoS ONE}
\bvolume{10}
\bpages{e0126346}.
\end{barticle}
%

\bptok{imsref}%
\endbibitem

\bibitem{EngelNagel2000}
%
\begin{bbook}[mr]
\bauthor{\bsnm{Engel},~\bfnm{Klaus-Jochen}\binits{K.-J.}} \AND
\bauthor{\bsnm{Nagel},~\bfnm{Rainer}\binits{R.}}
(\byear{2000}).
\btitle{One-Parameter Semigroups for Linear Evolution Equations}.
\bseries{Graduate Texts in Mathematics}
\bvolume{194}.
\blocation{New York}:
\bpublisher{Springer}.
\bid{mr={1721989}}
\end{bbook}
%

\bptok{imsref}%
\endbibitem

\bibitem{EnglerSenSweeneyDisher2006}
%
\begin{barticle}[auto:parserefs-M02]
\bauthor{\bsnm{Engler},~\bfnm{A.~J.}\binits{A.J.}},
\bauthor{\bsnm{Sen},~\bfnm{S.}\binits{S.}},
\bauthor{\bsnm{Sweeney},~\bfnm{H.~L.}\binits{H.L.}} \AND
\bauthor{\bsnm{Discher},~\bfnm{D.~E.}\binits{D.E.}}
(\byear{2006}).
\btitle{Matrix elasticity directs stem cell lineage specification}.
\bjournal{Cell}
\bvolume{126}
\bpages{677--689}.
\end{barticle}
%

\bptok{imsref}%
\endbibitem

\bibitem{EthierKurtz2009}
%
\begin{bbook}[mr]
\bauthor{\bsnm{Ethier},~\bfnm{Stewart~N.}\binits{S.N.}} \AND
\bauthor{\bsnm{Kurtz},~\bfnm{Thomas~G.}\binits{T.G.}}
(\byear{1986}).
\btitle{Markov Processes: Characterization and Convergence}.
\bseries{Wiley Series in Probability and Mathematical Statistics:
Probability and Mathematical Statistics}.
\blocation{New York}:
\bpublisher{Wiley}.
\bid{doi={10.1002/9780470316658}, mr={0838085}}
\bptnote{check year}%
\end{bbook}
%

\bptok{imsref}%
\endbibitem

\bibitem{FisherMarron2001}
%
\begin{barticle}[mr]
\bauthor{\bsnm{Fisher},~\bfnm{N.~I.}\binits{N.I.}} \AND
\bauthor{\bsnm{Marron},~\bfnm{J.~S.}\binits{J.S.}}
(\byear{2001}).
\btitle{Mode testing via the excess mass estimate}.
\bjournal{Biometrika}
\bvolume{88}
\bpages{499--517}.
\bid{doi={10.1093/biomet/88.2.499}, issn={0006-3444}, mr={1844848}}
\end{barticle}
%

\bptok{imsref}%
\endbibitem

\bibitem{Ghrist2008}
%
\begin{barticle}[mr]
\bauthor{\bsnm{Ghrist},~\bfnm{Robert}\binits{R.}}
(\byear{2008}).
\btitle{Barcodes: The persistent topology of data}.
\bjournal{Bull. Amer. Math. Soc. (N.S.)}
\bvolume{45}
\bpages{61--75}.
\bid{doi={10.1090/S0273-0979-07-01191-3}, issn={0273-0979}, mr={2358377}}
\end{barticle}
%

\bptok{imsref}%
\endbibitem

\bibitem{GoodGaskins1980}
%
\begin{barticle}[mr]
\bauthor{\bsnm{Good},~\bfnm{I.~J.}\binits{I.J.}} \AND
\bauthor{\bsnm{Gaskins},~\bfnm{R.~A.}\binits{R.A.}}
(\byear{1980}).
\btitle{Density estimation and bump-hunting by the penalized
likelihood method exemplified by scattering and meteorite data}.
\bjournal{J. Amer. Statist. Assoc.}
\bvolume{75}
\bpages{42--73}.
\bid{issn={0003-1291}, mr={0568579}}
\bptnote{check related, check pages}%
\end{barticle}
%

\bptok{imsref}%
\endbibitem

\bibitem{hallminottezhang04}
%
\begin{barticle}[mr]
\bauthor{\bsnm{Hall},~\bfnm{Peter}\binits{P.}},
\bauthor{\bsnm{Minnotte},~\bfnm{Michael~C.}\binits{M.C.}} \AND
\bauthor{\bsnm{Zhang},~\bfnm{Chunming}\binits{C.}}
(\byear{2004}).
\btitle{Bump hunting with non-{G}aussian kernels}.
\bjournal{Ann. Statist.}
\bvolume{32}
\bpages{2124--2141}.
\bid{doi={10.1214/009053604000000715}, issn={0090-5364}, mr={2102505}}
\end{barticle}
%

\bptok{imsref}%
\endbibitem

\bibitem{HartiganHartigan1985}
%
\begin{barticle}[mr]
\bauthor{\bsnm{Hartigan},~\bfnm{J.~A.}\binits{J.A.}} \AND
\bauthor{\bsnm{Hartigan},~\bfnm{P.~M.}\binits{P.M.}}
(\byear{1985}).
\btitle{The dip test of unimodality}.
\bjournal{Ann. Statist.}
\bvolume{13}
\bpages{70--84}.
\bid{doi={10.1214/aos/1176346577}, issn={0090-5364}, mr={0773153}}
\bptnote{check volume}%
\end{barticle}
%

\bptok{imsref}%
\endbibitem

\bibitem{HeoGambleKim2012}
%
\begin{barticle}[mr]
\bauthor{\bsnm{Heo},~\bfnm{Giseon}\binits{G.}},
\bauthor{\bsnm{Gamble},~\bfnm{Jennifer}\binits{J.}} \AND
\bauthor{\bsnm{Kim},~\bfnm{Peter~T.}\binits{P.T.}}
(\byear{2012}).
\btitle{Topological analysis of variance and the maxillary complex}.
\bjournal{J. Amer. Statist. Assoc.}
\bvolume{107}
\bpages{477--492}.
\bid{doi={10.1080/01621459.2011.641430}, issn={0162-1459}, mr={2980059}}
\end{barticle}
%

\bptok{imsref}%
\endbibitem

\bibitem{Karlin1968}
%
\begin{bbook}[mr]
\bauthor{\bsnm{Karlin},~\bfnm{Samuel}\binits{S.}}
(\byear{1968}).
\btitle{Total Positivity. {V}ol. {I}}.
\blocation{Stanford, CA}:
\bpublisher{Stanford Univ. Press}.
\bid{mr={0230102}}
\end{bbook}
%

\bptok{imsref}%
\endbibitem

\bibitem{Klemela2006}
%
\begin{barticle}[mr]
\bauthor{\bsnm{Klemel{\"a}},~\bfnm{Jussi}\binits{J.}}
(\byear{2006}).
\btitle{Visualization of multivariate density estimates with shape trees}.
\bjournal{J. Comput. Graph. Statist.}
\bvolume{15}
\bpages{372--397}.
\bid{doi={10.1198/106186006X113007}, issn={1061-8600}, mr={2256150}}
\end{barticle}
%

\bptok{imsref}%
\endbibitem

\bibitem{Lindeberg1994}
%
\begin{bbook}[auto:parserefs-M02]
\bauthor{\bsnm{Lindeberg},~\bfnm{T.}\binits{T.}}
(\byear{1994}).
\btitle{Scale-Space Theory in Computer Vision}.
\blocation{Boston}:
\bpublisher{Kluwer}.
\end{bbook}
%

\bptok{imsref}%
\endbibitem

\bibitem{Lindeberg2011}
%
\begin{barticle}[mr]
\bauthor{\bsnm{Lindeberg},~\bfnm{Tony}\binits{T.}}
(\byear{2011}).
\btitle{Generalized {G}aussian scale-space axiomatics comprising
linear scale-space, affine scale-space and spatio-temporal scale-space}.
\bjournal{J. Math. Imaging Vision}
\bvolume{40}
\bpages{36--81}.
\bid{doi={10.1007/s10851-010-0242-2}, issn={0924-9907}, mr={2782119}}
\end{barticle}
%

\bptok{imsref}%
\endbibitem

\bibitem{MairhuberSchoenbergWilliamson1959}
%
\begin{barticle}[mr]
\bauthor{\bsnm{Mairhuber},~\bfnm{J.~C.}\binits{J.C.}},
\bauthor{\bsnm{Schoenberg},~\bfnm{I.~J.}\binits{I.J.}} \AND
\bauthor{\bsnm{Williamson},~\bfnm{R.~E.}\binits{R.E.}}
(\byear{1959}).
\btitle{On variation diminishing transformations of the circle}.
\bjournal{Rend. Circ. Mat. Palermo (2)}
\bvolume{8}
\bpages{241--270}.
\bid{issn={0009-725X}, mr={0113109}}
\end{barticle}
%

\bptok{imsref}%
\endbibitem

\bibitem{mardiajupp00}
%
\begin{bbook}[mr]
\bauthor{\bsnm{Mardia},~\bfnm{Kanti~V.}\binits{K.V.}} \AND
\bauthor{\bsnm{Jupp},~\bfnm{Peter~E.}\binits{P.E.}}
(\byear{2000}).
\btitle{Directional Statistics}, \bedition{2nd}~ed.
\bseries{Wiley Series in Probability and Statistics}.
\blocation{Chichester}:
\bpublisher{Wiley}.
\bid{mr={1828667}}
\end{bbook}
%

\bptok{imsref}%
\endbibitem

\bibitem{Minnotte1997}
%
\begin{barticle}[mr]
\bauthor{\bsnm{Minnotte},~\bfnm{Michael~C.}\binits{M.C.}}
(\byear{1997}).
\btitle{Nonparametric testing of the existence of modes}.
\bjournal{Ann. Statist.}
\bvolume{25}
\bpages{1646--1660}.
\bid{doi={10.1214/aos/1031594735}, issn={0090-5364}, mr={1463568}}
\bptnote{check volume}%
\end{barticle}
%

\bptok{imsref}%
\endbibitem

\bibitem{MinnotteScott1993}
%
\begin{barticle}[auto:parserefs-M02]
\bauthor{\bsnm{Minnotte},~\bfnm{M.~C.}\binits{M.C.}} \AND
\bauthor{\bsnm{Scott},~\bfnm{D.~W.}\binits{D.W.}}
(\byear{1993}).
\btitle{The mode tree: A~tool for visualization of nonparametric
density features}.
\bjournal{J. Comput. Graph. Statist.}
\bvolume{2}
\bpages{51--68}.
\end{barticle}
%

\bptok{imsref}%
\endbibitem

\bibitem{MullerSawitzki1991}
%
\begin{barticle}[mr]
\bauthor{\bsnm{M{\"u}ller},~\bfnm{D.~W.}\binits{D.W.}} \AND
\bauthor{\bsnm{Sawitzki},~\bfnm{G.}\binits{G.}}
(\byear{1991}).
\btitle{Excess mass estimates and tests for multimodality}.
\bjournal{J. Amer. Statist. Assoc.}
\bvolume{86}
\bpages{738--746}.
\bid{issn={0162-1459}, mr={1147099}}
\end{barticle}
%

\bptok{imsref}%
\endbibitem

\bibitem{Munk1999}
%
\begin{barticle}[mr]
\bauthor{\bsnm{Munk},~\bfnm{Axel}\binits{A.}}
(\byear{1999}).
\btitle{Optimal inference for circular variation diminishing
experiments with applications to the von-{M}ises distribution and the
{F}isher--{E}fron parabola model}.
\bjournal{Metrika}
\bvolume{50}
\bpages{1--17}.
\bid{issn={0026-1335}, mr={1749579}}
\end{barticle}
%

\bptok{imsref}%
\endbibitem

\bibitem{OliveiraCrujeirasRodrigues-Casal2013}
%
\begin{barticle}[auto:parserefs-M02]
\bauthor{\bsnm{Oliveira},~\bfnm{M.}\binits{M.}},
\bauthor{\bsnm{Crujeiras},~\bfnm{R.~M.}\binits{R.M.}} \AND
\bauthor{\bsnm{Rodr{\'{\i}}guez-Casal},~\bfnm{A.}\binits{A.}}
(\byear{2013}).
\btitle{CircSiZer: An exploratory tool for circular data}.
\bjournal{Environmental and Ecological Statistics}
\bvolume{21}
\bpages{143--159}.
\end{barticle}
%

\bptok{imsref}%
\endbibitem

\bibitem{Ooi2002}
%
\begin{barticle}[mr]
\bauthor{\bsnm{Ooi},~\bfnm{Hong}\binits{H.}}
(\byear{2002}).
\btitle{Density visualization and mode hunting using trees}.
\bjournal{J. Comput. Graph. Statist.}
\bvolume{11}
\bpages{328--347}.
\bid{doi={10.1198/106186002760180545}, issn={1061-8600}, mr={1938139}}
\bptnote{check pages}%
\end{barticle}
%

\bptok{imsref}%
\endbibitem

\bibitem{Panaretos2013}
%
\begin{barticle}[mr]
\bauthor{\bsnm{Panaretos},~\bfnm{Victor~M.}\binits{V.M.}},
\bauthor{\bsnm{Pham},~\bfnm{Tung}\binits{T.}} \AND
\bauthor{\bsnm{Yao},~\bfnm{Zhigang}\binits{Z.}}
(\byear{2014}).
\btitle{Principal flows}.
\bjournal{J. Amer. Statist. Assoc.}
\bvolume{109}
\bpages{424--436}.
\bid{doi={10.1080/01621459.2013.849199}, issn={0162-1459}, mr={3180574}}
\bptnote{check volume, check pages, check year}%
\end{barticle}
%

\bptok{imsref}%
\endbibitem

\bibitem{Peetre1959}
%
\begin{barticle}[mr]
\bauthor{\bsnm{Peetre},~\bfnm{Jaak}\binits{J.}}
(\byear{1959}).
\btitle{Une caract\'erisation abstraite des op\'erateurs diff\'erentiels}.
\bjournal{Math. Scand.}
\bvolume{7}
\bpages{211--218}.
\bid{issn={0025-5521}, mr={0112146}}
\end{barticle}
%

\bptok{imsref}%
\endbibitem

\bibitem{Polonik1995}
%
\begin{barticle}[mr]
\bauthor{\bsnm{Polonik},~\bfnm{Wolfgang}\binits{W.}}
(\byear{1995}).
\btitle{Measuring mass concentrations and estimating density contour
clusters -- An excess mass approach}.
\bjournal{Ann. Statist.}
\bvolume{23}
\bpages{855--881}.
\bid{doi={10.1214/aos/1176324626}, issn={0090-5364}, mr={1345204}}
\bptnote{check volume}%
\end{barticle}
%

\bptok{imsref}%
\endbibitem

\bibitem{Rehfeldt2007}
%
\begin{barticle}[auto:parserefs-M02]
\bauthor{\bsnm{Rehfeldt},~\bfnm{F.}\binits{F.}},
\bauthor{\bsnm{Engler},~\bfnm{A.~J.}\binits{A.J.}},
\bauthor{\bsnm{Eckhardt},~\bfnm{A.}\binits{A.}},
\bauthor{\bsnm{Ahmed},~\bfnm{F.}\binits{F.}} \AND
\bauthor{\bsnm{Discher},~\bfnm{D.~E.}\binits{D.E.}}
(\byear{2007}).
\btitle{Cell responses to the mechanochemical microenvironment --
implications for regenerative medicine and drug delivery}.
\bjournal{Advanced Drug Delivery Reviews}
\bvolume{59}
\bpages{1329--1339}.
\end{barticle}
%

\bptok{imsref}%
\endbibitem

\bibitem{Sakai1996}
%
\begin{bbook}[mr]
\bauthor{\bsnm{Sakai},~\bfnm{Takashi}\binits{T.}}
(\byear{1996}).
\btitle{Riemannian Geometry}.
\bseries{Translations of Mathematical Monographs}
\bvolume{149}.
\blocation{Providence, RI}:
\bpublisher{Amer. Math. Soc}.
\bid{mr={1390760}}
\end{bbook}
%

\bptok{imsref}%
\endbibitem

\bibitem{SHMD13}
%
\begin{barticle}[mr]
\bauthor{\bsnm{Schmidt-Hieber},~\bfnm{Johannes}\binits{J.}},
\bauthor{\bsnm{Munk},~\bfnm{Axel}\binits{A.}} \AND
\bauthor{\bsnm{D{\"u}mbgen},~\bfnm{Lutz}\binits{L.}}
(\byear{2013}).
\btitle{Multiscale methods for shape constraints in deconvolution:
Confidence statements for qualitative features}.
\bjournal{Ann. Statist.}
\bvolume{41}
\bpages{1299--1328}.
\bid{doi={10.1214/13-AOS1089}, issn={0090-5364}, mr={3113812}}
\end{barticle}
%

\bptok{imsref}%
\endbibitem

\bibitem{SchwartzmanGavrilovAdler2011}
%
\begin{barticle}[mr]
\bauthor{\bsnm{Schwartzman},~\bfnm{Armin}\binits{A.}},
\bauthor{\bsnm{Gavrilov},~\bfnm{Yulia}\binits{Y.}} \AND
\bauthor{\bsnm{Adler},~\bfnm{Robert~J.}\binits{R.J.}}
(\byear{2011}).
\btitle{Multiple testing of local maxima for detection of peaks in 1{D}}.
\bjournal{Ann. Statist.}
\bvolume{39}
\bpages{3290--3319}.
\bid{doi={10.1214/11-AOS943}, issn={0090-5364}, mr={3012409}}
\end{barticle}
%

\bptok{imsref}%
\endbibitem

\bibitem{SchwartzmanJaffeGavrilovMeyer2013}
%
\begin{barticle}[mr]
\bauthor{\bsnm{Schwartzman},~\bfnm{Armin}\binits{A.}},
\bauthor{\bsnm{Jaffe},~\bfnm{Andrew}\binits{A.}},
\bauthor{\bsnm{Gavrilov},~\bfnm{Yulia}\binits{Y.}} \AND
\bauthor{\bsnm{Meyer},~\bfnm{Clifford~A.}\binits{C.A.}}
(\byear{2013}).
\btitle{Multiple testing of local maxima for detection of peaks in
{C}h{IP}-{S}eq data}.
\bjournal{Ann. Appl. Stat.}
\bvolume{7}
\bpages{471--494}.
\bid{doi={10.1214/12-AOAS594}, issn={1932-6157}, mr={3086427}}
\end{barticle}
%

\bptok{imsref}%
\endbibitem

\bibitem{Silverman1981}
%
\begin{barticle}[mr]
\bauthor{\bsnm{Silverman},~\bfnm{B.~W.}\binits{B.W.}}
(\byear{1981}).
\btitle{Using kernel density estimates to investigate multimodality}.
\bjournal{J. R. Stat. Soc. Ser. B. Stat. Methodol.}
\bvolume{43}
\bpages{97--99}.
\bid{issn={0035-9246}, mr={0610384}}
\bptnote{check volume}%
\end{barticle}
%

\bptok{imsref}%
\endbibitem

\bibitem{Taylor2008}
%
\begin{barticle}[mr]
\bauthor{\bsnm{Taylor},~\bfnm{Charles~C.}\binits{C.C.}}
(\byear{2008}).
\btitle{Automatic bandwidth selection for circular density estimation}.
\bjournal{Comput. Statist. Data Anal.}
\bvolume{52}
\bpages{3493--3500}.
\bid{doi={10.1016/j.csda.2007.11.003}, issn={0167-9473}, mr={2427365}}
\end{barticle}
%

\bptok{imsref}%
\endbibitem

\bibitem{WGS91}
%
\begin{barticle}[auto:parserefs-M02]
\bauthor{\bsnm{Wada},~\bfnm{T.}\binits{T.}},
\bauthor{\bsnm{Gu},~\bfnm{Y.~H.}\binits{Y.H.}} \AND
\bauthor{\bsnm{Sato},~\bfnm{M.}\binits{M.}}
(\byear{1991}).
\btitle{Scale-space filtering for periodic waveforms}.
\bjournal{Systems and Computers in Japan}
\bvolume{22}
\bpages{45--54}.
\end{barticle}
%

\bptok{imsref}%
\endbibitem

\bibitem{Weickert1998}
%
\begin{bbook}[mr]
\bauthor{\bsnm{Weickert},~\bfnm{Joachim}\binits{J.}}
(\byear{1998}).
\btitle{Anisotropic Diffusion in Image Processing}.
\bseries{European Consortium for Mathematics in Industry}.
\blocation{Stuttgart}:
\bpublisher{B. G. Teubner}.
\bid{mr={1666943}}
\end{bbook}
%

\bptok{imsref}%
\endbibitem

\bibitem{WeickertIshikawaAtuschi1999}
%
\begin{barticle}[mr]
\bauthor{\bsnm{Weickert},~\bfnm{Joachim}\binits{J.}},
\bauthor{\bsnm{Ishikawa},~\bfnm{Seiji}\binits{S.}} \AND
\bauthor{\bsnm{Imiya},~\bfnm{Atsushi}\binits{A.}}
(\byear{1999}).
\btitle{Linear scale-space has first been proposed in {J}apan}.
\bjournal{J. Math. Imaging Vision}
\bvolume{10}
\bpages{237--252}.
\bid{doi={10.1023/A:1008344623873}, issn={0924-9907}, mr={1695946}}
\end{barticle}
%

\bptok{imsref}%
\endbibitem

\bibitem{Zemel2010a}
%
\begin{barticle}[auto:parserefs-M02]
\bauthor{\bsnm{Zemel},~\bfnm{A.}\binits{A.}},
\bauthor{\bsnm{Rehfeldt},~\bfnm{F.}\binits{F.}},
\bauthor{\bsnm{Brown},~\bfnm{A.~E.~X.}\binits{A.E.X.}},
\bauthor{\bsnm{Discher},~\bfnm{D.~E.}\binits{D.E.}} \AND
\bauthor{\bsnm{Safran},~\bfnm{S.~A.}\binits{S.A.}}
(\byear{2010}).
\btitle{Optimal matrix rigidity for stress-fibre polarization in stem cells}.
\bjournal{Nat. Phys.}
\bvolume{6}
\bpages{468--473}.
\end{barticle}
%

\bptok{imsref}%
\endbibitem
\end{thebibliography}
\end{document}